\documentclass[11pt]{article}

\usepackage[colorlinks,allcolors=blue]{hyperref}
\usepackage{amssymb,amsmath,amsthm,mathrsfs}
\usepackage{venturis2}
\usepackage[T1]{fontenc}
\usepackage{comment}
\usepackage{graphicx, tikz, tikz-3dplot} 
\usetikzlibrary{arrows.meta}

\usepackage{wasysym}
\usepackage{color}
\usepackage{graphicx}	

\usepackage{amssymb,amsmath,amsthm}
\usepackage{mathrsfs}
\usepackage[colorlinks,allcolors=blue]{hyperref}
\usepackage{graphicx}
\usepackage{amsfonts,amscd}
\usepackage[all,cmtip]{xy}
\usepackage{cleveref}
\usepackage{xcolor}
\usepackage{arydshln}
\usepackage{comment}
\usepackage{cases}
\usepackage{tikz-3dplot}
\definecolor{mg}{rgb}  {0.85, 0.,  0.85}

\def\a{\mathfrak a}
\def\AA{\mathcal A}
\def\B{\mathscr B}
\def\C{\mathbb C}
\def\CC{\mathcal C}
\def\d{\mathrm d}

\def\Drond{\mathscr D}

\def\EE{\mathscr E}
\def\EEC{\mathscr E_{\mathrm C}}
\def\E{\mathcal  E}
\def\f{\mathfrak f}
\def\F{\mathscr F}
\def\g{\mathfrak g}
\def\G{\mathcal G}
\def\H{\mathcal H}
\def\Hrond{\mathscr H}
\def\h{\mathfrak h}
\def\K{\mathscr K}

\def\L{\mathcal L}

\def\N{\mathbb N}
\def\NN{\mathrm N}

\def\P{\mathcal P}
\def\QQ{\mathcal Q}
\def\R{\mathbb R}
\def\S{\mathbb S}
\def\s{\mathfrak s}
\def\T{\mathcal T}

\def\u{\mathfrak u}
\def\U{\mathcal U}

\def\v{\mathfrak v}
\def\V{\mathcal V}
\def\W{\mathcal W}

\def\Z{\mathbb Z}
  
\def\DD{\sqrt{-\Delta}}

\def\Tr{\mathrm{Tr}}

\def\diag{\mathop{\mathrm{diag}}\nolimits}
\def\e{\mathop{\mathrm{e}}\nolimits}

\def\Ran{\mathop{\mathrm{Ran}}\nolimits}

\def\sgn{\mathop{\mathrm{sgn}}\nolimits}

\DeclareMathOperator*{\slim}{s\hspace{0.1pt}-\hspace{0.1pt}lim}
\def\supp{\mathop{\mathrm{supp}}\nolimits}
\def\lone{\mathop{\mathrm{L}^1}\nolimits}
\def\ltwo{\mathop{\mathrm{L}^2}\nolimits}

\def\arctanh{\mathop{\mathrm{arctanh}}\nolimits}

\def\sech{\mathop{\mathrm{sech}}\nolimits}
\def\sp{\mathop{\mathrm{span}}\nolimits}
\def\tl{\tilde \lambda^\theta}
\def\tP{\tilde{\mathcal{P}}^\theta}

\def\tH{{{\tilde \Hrond}^\theta}}

\def\Var{\mathrm{Var}}
\def\ind{\mathop{\mathrm{ind}}\nolimits}
\def\q{\mathfrak{q}}

\newtheorem{Theorem}{Theorem}[section]
\newtheorem{Remark}[Theorem]{Remark}
\newtheorem{Lemma}[Theorem]{Lemma}

\newtheorem{Proposition}[Theorem]{Proposition}
\newtheorem{Definition}[Theorem]{Definition}

\begin{document}


\title{Topological Levinson's theorem and corrections at thresholds: \\
the full picture in a quasi-1D example}

\author{T.~T. Nguyen${}^1$, D. Parra${}^2$, S. Richard${}^3$\footnote{Partially supported by the JSPS Grant-in-Aid for scientific research C no 21K03292.}}

\date{\small}
\maketitle
\vspace{-1cm}

\begin{quote}
\begin{itemize}
\item[1] Department of physics, Graduate School of Science, Nagoya University, Furo-cho, Chikusa-ku, Nagoya, 464-8602, Japan
\item[2] 
Departamento de Matem\'atica y Estad\'istica, Universidad de La Frontera, \\ Av. Francisco Salazar 01145, Temuco, Chile
\item[3] Institute for Liberal Arts and Sciences \& Graduate School of Mathematics, Nagoya University, Furo-cho, Chikusa-ku, Nagoya, 464-8601, Japan
\item[] E-mail:  nguyen.tue.tai.e0@s.mail.nagoya-u.ac.jp, daniel.parra@ufrontera.cl, \\ 
${}\qquad \quad{}$richard@math.nagoya-u.ac.jp
\end{itemize}
\end{quote}

\date{\small}
\maketitle


\begin{abstract}
Various threshold effects are investigated on
a discrete quasi-1D scattering system. In particular, one of these effects is 
to add corrections to Levinson's theorem.
We explain how these corrections are due to the opening or to the closing 
of channels of scattering, and how these contributions can be computed as partial winding numbers on newly introduced operators.
Embedded thresholds, thresholds associated with changes of spectral multiplicity, 
and doubly degenerate thresholds are exhibited and analyzed.
Most of the investigations are of an analytical nature, but the final equalities rely on 
a $C^*$-algebraic framework.
\end{abstract}


\smallskip



\section{Introduction}\label{sec:intro}
\setcounter{equation}{0}

Levinson’s theorem is a relation between the number of bound states of a quantum
mechanical system and an expression arising from the scattering part of the system.
The original formulation by N. Levinson in \cite{Lev} was given in the context of
a Schr\"odinger operator with a spherically symmetric potential.
It has then been extended and refined in numerous papers and in various contexts. 
A common feature of these results is the appearance of correction terms
and of a regularization process on the scattering side of the equality. 
Regularizations are often  related to
the infinite dimension of all fixed energy subspaces, and can
be disregarded by considering quasi-1D systems.
On the other hand, corrections are usually related
to thresholds at the endpoints of the spectrum of the underlying
unperturbed operator, and are explicitly computed on an ad hoc basis.
Since it is impossible to mention all paper containing a Levinson's theorem 
with correction terms,
let us just mention two review papers \cite{Ma,Ri} which contain several
earlier references.

In the discrete setting, \cite{HKS} is one of the earliest works on Levinson's theorem.
For more general graphs, let us mention \cite{BSB,ChildI, ChildII}, 
while for recent results in the quasi-one dimensional setting include \cite{Bal1, Bal2, Bal3,IT19}.  In higher dimension, \cite{SB16} contains an extension of Levinson's theorem based on the notion of surface states. However, none of the systems considered in these references 
possess embedded thresholds, and hence the only effect considered is of endpoint thresholds.
It is only in the recent publication \cite{APRR} that systems with embedded thresholds
were considered for the first time, with their importance for Levinson’s theorem being highlighted.
However, even in this reference some restrictions were imposed, leading only to a subfamily of threshold's behaviors. The current manuscript removes all restrictions, provides the full picture, and exhibits new phenomena.

Let us be more concrete about the content of the paper and its specificity.
The model we consider has been introduced and partially studied in \cite{NRT}. 
In this reference it appears as parts of a discrete model of scattering theory in a half-space,
see also \cite{Cha00,CS00,Fra03,Fra04,JL00,RT16,RT17} for systems leading to a similar decomposition. 
Independently, the model corresponds to a magnetic Laplace operators acting on the graph $\N\times (\Z/N \Z)$ for some $N\in \N$, corresponding to a discrete half-cylinder, see Figure \ref{fig:cylinder} and Remark \ref{rem:magnetic}.
\begin{figure}[h]
\centering
    \tdplotsetmaincoords{80}{-25} 
    \begin{tikzpicture}[tdplot_main_coords,scale=2]
    \def\R{1}      
    \def\L{5}      
    \def\nz{9}     
    \def\nt{10}    
    \def\dotrad{1pt} 
    
    \tikzset{
      meshedge/.style={line width=0.3pt, draw=black!70},
      meshedgeDashed/.style={line width=0.3pt, draw=black!70, dashed}
    }
    
    \definecolor{c1}{rgb}{1,0,0}
    \definecolor{c2}{rgb}{1,0.5,0}
    \definecolor{c3}{rgb}{1,1,0}
    \definecolor{c4}{rgb}{0,1,0}
    \definecolor{c5}{rgb}{0,1,1}
    \definecolor{c6}{rgb}{0,0,1}
    \definecolor{c7}{rgb}{0.5,0,1}
    \definecolor{c8}{rgb}{1,0,1}
    \definecolor{c9}{rgb}{0.6,0.2,0.2}
    \definecolor{c10}{rgb}{0.2,0.6,0.2}
    \definecolor{c11}{rgb}{0.2,0.2,0.6}
    \definecolor{c12}{rgb}{0.6,0.6,0.2}
    \def\mycolors{{c1}{c2}{c3}{c4}{c5}{c6}{c7}{c8}{c9}{c10}{c11}{c12}}
    
    \pgfmathtruncatemacro{\nzm}{\nz-2}
    \pgfmathtruncatemacro{\ntm}{\nt-1}
    \pgfmathsetmacro{\dx}{\L/(\nz-1)}
    
    \foreach \i in {0,...,\nzm}{
      \pgfmathsetmacro{\xA}{\i*\dx}
      \pgfmathsetmacro{\xB}{(\i+1)*\dx}
      \pgfmathtruncatemacro{\rem}{\nz-1-\i}
      \ifnum\rem<2
        \def\thisstyle{meshedgeDashed}
      \else
        \def\thisstyle{meshedge}
      \fi
      \foreach \j in {0,...,\ntm}{
        \pgfmathsetmacro{\angA}{360*\j/\nt}
        \pgfmathsetmacro{\angB}{360*(\j+1)/\nt}
        \path (\xA,{\R*cos(\angA)},{\R*sin(\angA)}) coordinate (P1)
              (\xB,{\R*cos(\angA)},{\R*sin(\angA)}) coordinate (P2)
              (\xB,{\R*cos(\angB)},{\R*sin(\angB)}) coordinate (P3)
              (\xA,{\R*cos(\angB)},{\R*sin(\angB)}) coordinate (P4);
        \ifnum\rem<2
          \draw[\thisstyle] (P1)--(P2);
        \else
          \draw[\thisstyle] (P1)--(P2)--(P3)--(P4)--cycle;
        \fi
      }
    }

    \foreach \j in {0,...,\numexpr\nt-1} {
      \pgfmathsetmacro{\ang}{360*\j/\nt}
      \pgfmathsetmacro{\xx}{0}
      \pgfmathsetmacro{\yy}{\R*cos(\ang)}
      \pgfmathsetmacro{\zz}{\R*sin(\ang)}
      \pgfmathparse{{"c\the\numexpr\j+1\relax"}}
      \edef\thiscolor{\pgfmathresult}
      \filldraw[\thiscolor, draw=none] (\xx,\yy,\zz) circle (\dotrad);
    }
    \pgfmathsetmacro{\xi}{2.8*\dx}
    \pgfmathsetmacro{\xf}{7.3*\dx}
    \draw[-stealth, ultra thick] (\xi,0,-1.2) -- (\xf,0,-1.2);
    \end{tikzpicture}
\caption{Discrete cylinder with perturbations located on its first layer.
The direction of the magnetic field is indicated by the arrow, with an intensity
parametrized by $\theta \in [0,2\pi)$.}
\label{fig:cylinder}
\end{figure}
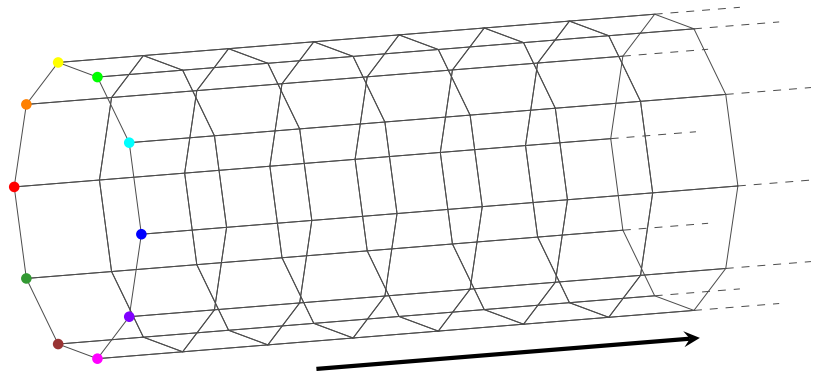
The magnetic field is constant in the $\N$ direction, and is parameterized by $\theta\in [0,2\pi)$.
We refer to \cite{CdV,KS} for seminal references on discrete magnetic operators.
Perturbations are then added on the boundary of the system, namely
they are supported on $\{0\}\times (\Z/N \Z)$.
These systems have changes of multiplicity in the continuous spectrum, 
and accordingly embedded thresholds. The set
of thresholds is defined by the set $\{\lambda^\theta_j\pm 2\}_{j=1}^N$, 
where $\{\lambda^\theta_j\}_{j=1}^N$ correspond to the eigenvalues of a magnetic Laplace operator on $\Z/N\Z$. 
These thresholds naturally lead to the structure of a Hilbert space with fibers of finite and locally constant dimension, in which most of the analysis has to be performed.
This setup is precisely introduced in Section \ref{sec_ini}, where also basics spectral
and scattering results are recalled from \cite{NRT}.

Levinson's theorem is then obtained as an index theorem in scattering theory.
The general framework has been exhibited in \cite{KR06, KR07, KR08}, and then summarized in the review paper \cite{Ri}. However, all examples presented in these references
were based on systems in $\R^d$ and not in a discrete setting.
The key idea is to construct a $C^*$-algebraic setting which accommodates one of the wave operators (no need to consider both). Once in a suitable $C^*$-algebra, its Fredholm index can be evaluated by its image in a sub-algebra of the Calkin algebra. More precisely, by looking at the image of the wave operators in a quotient algebra (obtained by quotienting it with the ideal of compact operators) one obtains a unitary operator. 
This operator consists of several pieces, among which the main one is the scatting operator $S^\theta$. 
The other pieces can be interpreted as operators of opening of new channels of scattering, or operators of closing of channels of scattering. All of them are located at thresholds, and the their precise form depend on the nature of the thresholds. Note that the scattering operator is fibered in the energy variable, while these newly introduced operator are fibered in a conjugate variable. Once a winding number
is suitably applied to this family of operators, the already mentioned Fredholm index is recovered, and the equality of these two quantities is precisely Levinson's theorem.

In the paper \cite{APRR}, the intensity of the magnetic field was restricted to $(0,\pi)$.
As a consequence, all internal thresholds correspond to a change by $\pm 1$ of the spectral multiplicity: $+1$ for the opening of a new channel of scattering, $-1$ for the closing of a channel of scattering. At these thresholds, the scattering matrix restricted to this specific channel can only take two values: $\pm 1$,  the value $-1$ being generic while the value $1$ being exceptional. This situation was fully investigated in the mentioned reference
and led to a Levinson's theorem of the form:
\begin{equation}\label{eq_11}
\Var \big(\lambda \to \det S^\theta(\lambda)\big) +N-\frac{\#\big\{j\mid \s^\theta_{jj}(\lambda^\theta_j\pm 2)=1\big\}}2
=\# \sigma_{\rm p}(H^\theta)
\end{equation}
where $\s_{jj}^\theta(\lambda_j^\theta\pm 2)$ corresponds to a distinguished entry of the scattering matrix at the threshold energy $\lambda_j^\theta\pm 2$, and where $\Var \big(\lambda \mapsto \det S^\theta(\lambda)\big)$ 
is  the total variation of the argument of the piecewise continuous function $\lambda\mapsto\det S^\theta(\lambda)$. 
On the r.h.s.~the expression $\# \sigma_{\rm p}(H^\theta)$
corresponds to the number of bound states of the perturbed system. 
Let us emphasize that the second and third term on the l.h.s.~of the above equality
are obtained as the partial winding number (or equivalently the partial variation of the argument) of the newly exhibited and already mentioned fibered operators of opening or of closing of a channel of scattering. 

In the present paper, we consider arbitrary intensity of the magnetic field in $[0,2\pi)$. This leads to two main new phenomena. Firstly, for $\theta \in \{0,\pi\}$ the change of spectral
multiplicity belongs to $\{-2,-1,1,2\}$, depending on the thresholds. The change of multiplicity $\pm 2$ leads to new operators of opening or closing of channels of scattering
which are described in Section \ref{sec:C*} and the main result read then
\begin{equation}
\Var \big(\lambda \to \det S^\theta(\lambda)\big) + N - \tfrac{1}{2}C
 = \# \sigma_{\rm p}(H^\theta),
\end{equation}
where 
$$
C:=\#\big\{k\mid \tilde \s^\theta_{k}(\tilde \lambda^\theta_k\pm 2)=1 \big\} 
+ 2\#\big\{k\mid \tilde \s^\theta_{k}(\tilde \lambda^\theta_k\pm 2)= \left(\begin{smallmatrix}
1& 0 \\  0 & 1
\end{smallmatrix}\right)\big\}
+\#\Big\{k\mid \tilde \s^\theta_{k}(\tilde \lambda^\theta_k\pm 2)=\left(\begin{smallmatrix}
a& b \\  \overline{b} & -a \end{smallmatrix}\right) \Big\}.
$$
where $\tilde \s_{k}^\theta(\tilde \lambda_k^\theta\pm 2)$ corresponds now to the distinguished part of the scattering matrix at the threshold energy $\tilde \lambda_k^\theta\pm 2$. Note that a new indexation had to be introduced since
some eigenvalues $\lambda^\theta_j$ are now equal, while they are all distinct for
$\theta \not \in \{0,\pi\}$. Again all correction terms in the factor $C$
are explicitly computed as partial winding numbers of some fibered operators.

The second new phenomenon has been called the \emph{intricate case} and appears
only for $\theta =0$ for very specific perturbations. This special situation was mentioned in \cite{NRT} under the name \emph{degenerate case} because it was not possible to treat it at that time. The name has been changed because it corresponds to a surprising new effect, and a precise definition is given in Definition \ref{def_intricate}.
When $\theta=0$ and $N$ is even, 
the point $0$ in the spectrum of the unperturbed operator
corresponds to a doubly degenerate threshold: it is the closing of a channel of scattering, and simultaneously the opening of a new channel of scattering. 
In this specific situation and for very specific potential (with at least $N/2$ trivial values) a term in the wave operators which is compact for all other perturbations, all other $\theta$ and $N$, is suddenly no more compact. As a consequence, this term provides a new contribution when one considers the image of the wave operators through the ideal of compact operators.
At the algebraic level, the $C^*$-algebra introduced in Section \ref{subsecC} whose quotient is carefully studied in Section \ref{sec:C*} is suddenly no more sufficient.
An extended framework has to be introduced for this specific situation, and lead to a surprising connection between energies $-4$, $0$, and $4$.

This new situation is thoroughly studied in Section \ref{sec_Degenerate},
firstly for $N=2$ and then for a general even number $N$. 
These investigations impose the use of a larger algebra whose construction is borrowed
from \cite[Sec.~V.7]{Cordes} and which was used for the first time in this context in \cite{ANRR}.
Once in this framework, the topological version of Levinson's theorem appears quite naturally, and in this context it reads:
\begin{equation}
\Var \big(\lambda \to \det S^0(\lambda)\big) + N - \tfrac{1}{2}C-\tfrac{1}{2}
 = \# \sigma_{\rm p}(H^0),
\end{equation}
where 
$$
C:=\#\big\{k\mid \tilde \s^0_{k}(\tilde \lambda^0_k\pm 2)=1 \big\} 
+ 2\#\big\{k\mid \tilde \s^0_{k}(\tilde \lambda^0_k\pm 2)= \left(\begin{smallmatrix}
1& 0 \\  0 & 1
\end{smallmatrix}\right)\big\}
+\#\Big\{k\mid \tilde \s^0_{k}(\tilde \lambda^0_k\pm 2)=\left(\begin{smallmatrix}
a& b \\  \overline{b} & -a \end{smallmatrix}\right) \Big\}.
$$
The additional factor $\frac{1}{2}$ in the above formula is due to the intricate connection
between energies $-4$, $0$, and $4$. Let us mention another special effect of the doubly degenerate threshold at $0$
in the intricate case: 
even if the closure of one channel of scattering and simultaneously the opening of a new channel of scattering correspond to two changes of $\pm 1$ for the spectral multiplicity, the scattering matrix restricted to these specific channels do not take the values $\pm 1$, as it could be anticipated, 
it takes the values $\pm i$. When computing
the partial winding number of the corresponding operator, this provides the difference of factor $\frac{1}{2}$.

In summary, this paper completes and provides a full picture of the partial results obtained
in \cite{APRR} and \cite{NRT}. The correction terms appearing in Levinson's theorem
are now fully explained in all possible cases, and appear as partial winding numbers of new operators. 
Once all of them are taken into account, the topological nature of Levinson's theorem is restored.

\section{Model and existing results}\label{sec:themodel}
\setcounter{equation}{0}

In this section, we introduce the quantum system and give a rather complete description of the results obtained in \cite{NRT}. Part of the algebraic framework provided in \cite{APRR}
is also recalled. We present this material for arbitrary $\theta\in [0,2\pi)$.

\subsection{Spectral and scattering theory}\label{sec_ini}

The following exposition is partially based on the short review paper \cite{RiRIMS}.
We set $\N:=\{0,1,2,\dots\}$. In the Hilbert space $\ell^2(\N)$
we consider  the discrete Neumann
adjacency operator whose action on $\phi\in\ell^2(\N)$ is described by
$$
\big(\Delta_\NN\;\!\phi\big)(n)=
\begin{cases}
2^{1/2}\;\!\phi(1) &\hbox{if $n=0$}\\
2^{1/2}\;\!\phi(0)+\phi(2) &\hbox{if $n=1$}\\
\phi(n+1)+\phi(n-1) &\hbox{if $n\ge2$.}
\end{cases}
$$
For any fixed $N\in \N$ with $N\geq 2$ and for any fixed $\theta \in [0,2\pi)$ 
we also consider the $N\times N$ Hermitian matrix
\begin{equation*}
A^\theta:=
\begin{pmatrix}
0 & 1 & 0 &\cdots & 0 &\e^{-i\theta}\\
1 & 0 & 1 &\ddots &  & 0\\
0 & 1 &\ddots &\ddots &\ddots &\vdots\\
\vdots &\ddots &\ddots &\ddots & 1 & 0\\
0 & &\ddots & 1 & 0 & 1\\
\e^{i\theta} & 0 &\cdots & 0 & 1 & 0
\end{pmatrix}.
\end{equation*}
Note that in the special case $N=2$, this matrix takes the form 
$
A^\theta
=\left(\begin{smallmatrix}
0 & 1+\e^{-i\theta}\\
1+\e^{i\theta} & 0
\end{smallmatrix}\right).
$
These two ingredients lead to the self-adjoint operator $H_0(\theta)$ 
acting on $\ell^2\big(\N;\C^N\big)$ as
$$
H_0(\theta):=\Delta_\NN \otimes 1_N+ A^\theta
$$
with $1_N$ the $N \times N $ identity matrix.

This operator can be viewed as a discrete magnetic adjacency operator, see Remark \ref{rem:magnetic}. It is also the one which appears in a direct integral decomposition of 
an operator acting on $\ell^2(\N\times \Z)$ through a Bloch-Floquet
transformation, see \cite[Sec.~2]{NRT}. We do not emphasize this decomposition 
in the present work, but concentrate on each individual operator living on the fibers
of a direct integral.

The perturbed operator $H(\theta)$ describing the discrete quantum model is then given by
\begin{equation*}
H(\theta):=H_0(\theta)+V,
\end{equation*}
where $V$ is the multiplication operator by a nonzero, matrix-valued function
with support on $\{0\}\in \N$. In other words, there exists a nonzero function
$v:\{1,\dots,N\}\to\R$ such that for $\psi\in \ell^2\big(\N;\C^N\big)$, $n\in\N$, and $j\in \{1, \dots, N\}$
one has
$$
\big(H(\theta)\psi\big)_j(n)=\big(H_0(\theta)\psi\big)_j(n)+\delta_{0,n}\;\!v(j)\;\!\psi_j(0),
$$
with $\delta_{0,n}$ the Kronecker delta function.

\begin{Remark}\label{rem:magnetic}
$H_0(\theta)$ can be interpreted as a magnetic adjacency operator on the Cayley graph of the semi-group $\N\times (\Z/N \Z)$ with respect to a magnetic field constant in the $\N$ direction. In particular, it corresponds to choosing a magnetic potential supported only on edges of the form $\big((n,N),(n,1)\big)$ for any $n\in \N$. In this representation, the perturbation $V$ corresponds
to a multiplicative perturbation with support on $\{0\}\times (\Z/N \Z)$.
\end{Remark}

We now set
$$
\h:=\ltwo\big([0,\pi),\tfrac{\d\omega}\pi;\C^N\big)
$$
and define a transformation in the $\N$-variable
$$
\G: \ell^2\big(\N;\C^N\big) \to\h
$$
given for $\psi \in \ell^2\big(\N;\C^N\big)_{\rm fin}$ (the finitely supported functions), $\omega\in [0,\pi)$, and $j\in \{1,\dots,N\}$ by
$$
\big(\G \psi\big)_j(\omega)
:=2^{1/2}\sum_{n\ge1}\cos(n\;\!\omega)\psi_j(n)+\psi_j(0).
$$
The transformation $\G$ extends to a unitary operator from
$\ell^2\big(\N;\C^N\big)$ to $\h$, which we denote by the same
symbol, and a direct computation gives the equality
$$
\G\;\!H_0(\theta)\;\!\G^*= H^\theta_0
\quad\hbox{with}\quad
H^\theta_0:=2\cos(\Omega)\otimes 1_N+A^\theta
$$
with the operator $2\cos(\Omega)$ of multiplication by the function
$[0,\pi)\ni\omega\mapsto2\cos(\omega)\in\R$.

Through the same unitary transformation, the operator $H(\theta)$ 
is unitarily equivalent to 
\begin{equation*}
H^\theta:=2\cos(\Omega)\otimes 1_N+A^\theta+\diag(v)P_0
\end{equation*}
with
\begin{equation*}
\big(\diag(v)\;\!\f\big)_j
:=v(j)\;\!\f_j
\quad\hbox{and}\quad
\big(P_0\;\!\f\big)_j
:=\int_0^\pi \f_j(\omega)\tfrac{\d\omega}\pi
\end{equation*}
for $\f\in\h$, $j\in\{1,\dots,N\}$.
Observe that the term $\diag(v)P_0$ corresponds to a finite rank
perturbation, and therefore $H^\theta$ and $H_0^\theta$
differ only by a finite rank operator.

Let us now move to spectral results. 
A direct inspection shows that the matrix $A^\theta$ has eigenvalues
\begin{equation}\label{eq_eigenvalues}
\lambda_j^\theta:=2\cos\left(\frac{\theta+2\pi\;\!j}N\right),\quad j\in\{1,\dots,N\},
\end{equation}
with corresponding eigenvectors $\xi_j^\theta\in\C^N$ having components
$\big(\xi_j^\theta\big)_k:=\e^{i(\theta+2\pi j)k/N}$ for $j,k\in\{1,\dots,N\}$. Using the
notation $\P_j^\theta\equiv \tfrac{1}{N}|\xi^\theta_j\rangle \langle \xi^\theta_{j}|$ for the orthogonal projection associated to $\xi_j^\theta$, one has 
$A^\theta=\sum_{j=1}^N\lambda_j^\theta\;\!\P_j^\theta$.

The next step consists in exhibiting the spectral representation of $H_0^\theta$.
For that purpose, we first define the sets
\begin{equation*}
I_j^\theta:=\big(\lambda_j^\theta-2,\lambda_j^\theta+2\big)
\quad\hbox{and}\quad
I^\theta:=\cup_{j=1}^NI^\theta_j,
\end{equation*}
with $\lambda_j^\theta$ the eigenvalues of $A^\theta$. 
Also, we consider for $\lambda\in I^\theta$ the fiber Hilbert space
$$
\Hrond^\theta(\lambda)
:=\sp\big\{\P^\theta_j\C^N\mid\hbox{$j\in\{1,\dots,N\}$ such that
$\lambda\in I_j^\theta$}\big\}\subset\C^N,
$$
and the corresponding direct integral Hilbert space
\begin{equation}\label{eq_Hrond}
\Hrond^\theta:=\int_{I^\theta}^\oplus\Hrond^\theta(\lambda)\;\!\d\lambda.
\end{equation}
Then, the map $\F^\theta:\h\to\Hrond^\theta$ 
acting on $\f\in\h$ and for a.e.~$\lambda\in I^\theta$ as
$$
\big(\F^\theta \f\big)(\lambda)
:=\pi^{-1/2}\sum_{\{j\mid\lambda\in I_j^\theta\}}
\big(4-\big(\lambda-\lambda_j^\theta\big)^2\big)^{-1/4}\;\!\P^\theta_j
\f\left(\arccos\left(\tfrac{\lambda-\lambda_j^\theta}2\right)\right).
$$
is unitary.
In addition, $\F^\theta$ diagonalises the Hamiltonian $H_0^\theta$, namely for all
$\zeta\in\Hrond^\theta$ and a.e. $\lambda\in I^\theta$ one has
\begin{equation*}
\big(\F^\theta H_0^\theta\;\!(\F^{\theta})^*\zeta\big)(\lambda)
=\lambda\;\!\zeta(\lambda)
=\big(X^\theta\zeta\big)(\lambda),
\end{equation*}
with $X^\theta$ the (bounded) operator of multiplication by the variable in
$\Hrond^\theta$.
One infers  that
$H_0^\theta$ has purely absolutely continuous spectrum equal to
\begin{equation*}
\textstyle\sigma(H_0^\theta)
=\overline{\Ran(X^\theta)}
=\overline{I^\theta}
=\big[\big(\min_j\lambda_j^\theta\big)-2,\big(\max_j\lambda_j^\theta\big)+2\big]
\subset[-4,4].
\end{equation*}
Note that we use the notation  $\overline{I}$ for the closure of a set $I\subset \R$.
We also set $\sigma_{\rm p}(H^\theta)$ for the point spectrum of $H^\theta$.

The spectral representation of $H^\theta_0$ leads also naturally to the
notion of thresholds: these real values correspond to a change
of multiplicity of the spectrum.
Clearly, the set $\T^\theta$ of thresholds for the operator $H_0^\theta$
is given by
\begin{equation*}
\T^\theta:=\big\{\lambda_j^\theta\pm2\mid j\in\{1,\dots,N\}\big\}.
\end{equation*}

The main spectral result for $H^\theta$ has been presented in \cite[Prop.~1.2]{NRT},
and will be recalled when necessary. About scattering theory, since the difference $H^\theta-H^\theta_0$ is a finite rank
operator we observe that the wave operators
\begin{equation}\label{eq_WO}
W_{\pm}^\theta:=\slim_{t\to\pm\infty}\e^{itH^\theta}\e^{-itH^\theta_0}
\end{equation}
exist and are complete, see  \cite[Thm.~X.4.4]{Kat95}. 
As a consequence of the completeness, the scattering operator
$$
S^\theta:=(W_+^\theta)^*W_-^\theta
$$
is a unitary operator in $\h$ commuting with $H^\theta_0$, and thus $S^\theta$ is
decomposable in the spectral representation of $H^\theta_0$, that is
for $\zeta\in\Hrond^\theta$ and a.e.~$\lambda\in I^\theta$, one has
$$
\big(\F^\theta S^\theta (\F^\theta)^*\zeta\big)(\lambda)=S^\theta(\lambda)\zeta(\lambda),
$$
with the scattering matrix $S^\theta(\lambda)$ a unitary operator in $\Hrond^\theta(\lambda)$.
For simplicity, we shall often just write $S^\theta(X^\theta)$ for the operator 
$\F^\theta S^\theta (\F^\theta)^*$ since this operator corresponds to a matrix-valued multiplication operator.

Finally, for $j,j'\in\{1,\dots,N\}$ and for
$
\lambda\in\big(I^\theta_j\cap I^\theta_{j'}\big)
\setminus\big(\T^\theta\cup\sigma_{\rm p}(H^\theta)\big)
$
let us define the channel scattering matrix
$S^\theta(\lambda)_{j,j'}:=\P^\theta_jS^\theta(\lambda)\P^\theta_{j'}$
and consider the map
$$
\big(I^\theta_j\cap I^\theta_{j'}\big)
\setminus\big(\T^\theta\cup\sigma_{\rm p}(H^\theta)\big)
\ni\lambda\mapsto S^\theta(\lambda)_{j,j'}
\in\B\big(\P_{j'}^\theta\C^N;\P_j^\theta\C^N\big).
$$
The continuity of the scattering matrix at embedded eigenvalues has been shown in 
\cite[Thm.~3.10]{NRT}, while its behavior at thresholds has been studied in 
\cite[Thm.~3.9]{NRT}.
This latter result can be summarized as follows: For each $\lambda\in\T^\theta$, 
a channel can already be opened at the
energy $\lambda$ (in which case the existence and the equality of the
limits from the right and from the left is proved), it can open at the energy $\lambda$ (in
which case only the existence of the limit from the right is proved), or it can
close at the energy $\lambda$ (in which case only the existence of the
limit from the left is proved).

\subsection{C*-algebras}\label{subsecC}

We now briefly sketch the algebraic framework which played a crucial role in \cite{APRR}
for getting an index theorem in scattering theory. 
In this reference, the case $\theta\in (0,\pi)$ was thoroughly investigated 
(and it applies to $\theta \in (\pi,2\pi)$ as well). 
In this section we generalize the construction for $\theta \in [0,2\pi)$.
We firstly exhibit an algebra $\AA^\theta$ of multiplication operators acting on $\Hrond^\theta$.

\begin{Definition}\label{def_A_theta}
Let $\AA^\theta\subset \B\big(\Hrond^\theta\big)$ with  $a\in \AA^\theta$ if for any $\zeta \in \Hrond^\theta$ and $\lambda \in I^\theta$ one has
$$
\big[a\zeta\big](\lambda)\equiv \big[a(X^\theta)\zeta\big](\lambda)= a(\lambda) \zeta(\lambda)
$$
and for $j,j'\in \{1,\dots,N\}$ the maps
\begin{equation}\label{eq:mult}
I_j^\theta \cap I_{j'}^\theta \ni \lambda \mapsto a_{j,j'}(\lambda) :=
\P_j^\theta a(\lambda)\P_{j'}^\theta \in \B\big(\Hrond^\theta(\lambda)\big) \subset \B(\C^N)
\end{equation}
belongs to $C_0\big(I_j^\theta \cap I_{j'}^\theta; \B(\C^N)\big)$ if $\lambda_j^\theta\neq \lambda_{j'}^\theta$ while it belongs to $C\big(\overline{I_j^\theta}; \B(\C^N)\big)$ if $\lambda_j^\theta=\lambda_{j'}^\theta$.
\end{Definition}

In other words, the function defined in \eqref{eq:mult}
is a continuous function vanishing at the boundaries of the intersection $I_j^\theta \cap I_{j'}^\theta$ when these two intervals are not equal, and having limits at the boundaries of these intervals when they coincide. 

\begin{Remark}\label{rem:misleading}
In the sequel, we shall use the notation $a_{j,j'}$ for $\P_j^\theta a \P_{j'}^\theta$
even if this notation is slightly misleading. Indeed, $a_{j,j'}(\lambda)$ exists if and only if 
$\lambda \in  \overline{I_j^\theta \cap I_{j'}^\theta}$. 
If $\lambda \not \in 
\overline{I_j^\theta}$ or if $\lambda \not \in\overline{I^\theta_{j'}}$, then the expression containing
$a_{j,j'}(\lambda)$ should be interpreted as $0$.
More precisely, with the standard bra-ket notation, we shall write
$a_{j,j'}=\a_{j,j'}\otimes \;\!\tfrac{1}{N}|\xi^\theta_j\rangle \langle \xi^\theta_{j'}|$,
where $\a_{j,j'}$ denotes the scalar-valued multiplication operator by a function belonging to  
$C_0\big(I_j^\theta \cap I_{j'}^\theta; \C\big)$ if  $\lambda_j^\theta\neq \lambda_{j'}^\theta$, while it belongs to $C\big(\overline{I_j^\theta}; \C\big)$
if  $\lambda_j^\theta = \lambda_{j'}^\theta$.
Thus, with these notations and for any $\lambda \in I^\theta$ one has
\begin{equation}\label{eq_decc}
a(\lambda)=\sum_{\{j,j'\mid \lambda\in I^\theta_j, \lambda \in I^\theta_{j'}\}}
\a_{j,j'}(\lambda) \otimes \;\!\tfrac{1}{N}|\xi^\theta_j\rangle \langle \xi^\theta_{j'}|
\end{equation}
By the above abuse of notation, we shall conveniently write
$a=\sum_{j,j'}
\a_{j,j'}\otimes \;\!\frac{1}{N}|\xi^\theta_j\rangle \langle \xi^\theta_{j'}|$.
\end{Remark}

The interest of the algebra $\AA^\theta$ comes from the following statement, which can be directly inferred from \cite[Thms 3.9 \& 3.10]{NRT}\;\!:

\begin{Proposition}\label{prop_S}
For any $\theta \in [0,2\pi)$, the multiplication operator defined by the map $\lambda \to S^\theta(\lambda)$ belongs to $\AA^\theta$. 
\end{Proposition}

We now recall the definition of a useful unitary map, introduced in
\cite[Sec.~4.1]{NRT}.
For $j\in\{1,\dots,N\}$ we define the unitary map
$\V_j^\theta:\ltwo(I_j^\theta)\to\ltwo(\R)$ given on $\zeta\in\ltwo(I_j^\theta)$ and for $s\in\R$ by
\begin{equation}\label{eq:Vjt}
\big[\V_j^\theta\zeta\big](s)
:=\tfrac{2^{1/2}}{\cosh(s)}\;\!\zeta\big(\lambda_j^\theta+2\tanh(s)\big).
\end{equation}
Its adjoint $(\V^{\theta}_j)^*:\ltwo(\R)\to\ltwo(I_j^\theta)$ is then given on $f\in\ltwo(\R)$ and for  $\lambda\in I_j^\theta$ by
\begin{equation}\label{eq:Vjt*}
\big[(\V^{\theta}_j)^*f\big](\lambda)
=\left(\tfrac2{4-(\lambda-\lambda_j^\theta)^2}\right)^{1/2}
f\left(\arctanh\left(\tfrac{\lambda-\lambda_j^\theta}2\right)\right).
\end{equation}
Based on this, we define the unitary map
$\V^\theta:\Hrond^\theta\to\ltwo(\R;\C^N)$  acting on $\zeta\in\Hrond^\theta$ as
$$
\V^\theta\zeta
:=\sum_{j=1}^N\big(\V^\theta_j\otimes\P_j^\theta\big)\zeta|_{I^\theta_j},
$$
with adjoint $(\V^{\theta})^*:\ltwo(\R;\C^N)\to\Hrond^\theta$ acting on 
$f\in\ltwo(\R;\C^N)$ and for $\lambda\in I^\theta$ as
$$
\big[(\V^{\theta})^*f\big](\lambda)
=\sum_{\{j\mid\lambda\in I^\theta_j\}}
\left[\big((\V^{\theta}_j)^*\otimes\P_j^\theta\big)f\right](\lambda).
$$

Let us directly mentioned two useful results based on these unitary transformations. For the first one, we denote by $X$ the self-adjoint operator of multiplication by the variable in  $\ltwo(\R)$.

\begin{Lemma}[Lem.~3.5 in \cite{APRR}]\label{lem_tanh}
For  $\zeta\in\Hrond^\theta$ and $\lambda \in I^\theta$ one has
\begin{equation*}
\big[(\V^\theta)^* \big(\tanh(X)\otimes 1_N \big)\V^\theta \zeta\big](\lambda)
= \sum_{\{j\mid\lambda\in I_j^\theta\}}\left[\tfrac{X^\theta-\lambda^\theta_j}{2}\otimes \P_j^\theta \zeta\right](\lambda).
\end{equation*}
\end{Lemma}

Let also $D$ stand for the self-adjoint realization of the operator $-i\frac{\d}{\d x}$ in $\ltwo(\R)$. Clearly, $X$ and $D$ satisfy the canonical commutation relation. By functional calculus, we then define the bounded convolution operator $\eta(D)$ with $\eta\in C_0\big([-\infty,+\infty)\big)$, the algebra of continuous functions on $\R$ having a limit at $-\infty$ and vanishing at $+\infty$.
In $\ltwo(\R;\C^N)$ we finally introduce the family of operators 
$\eta(D)\otimes 1_N$
with $\eta\in  C_0\big([-\infty,+\infty)\big)$. 
These operators naturally generate a $C^*$-algebra $\CC$ isomorphic to  $C_0\big([-\infty,+\infty)\big)$.
Clearly, the operators $\frac{1}{2}\big(1-\tanh(\pi D)\big)\otimes 1_N$ and $\cosh(\pi D)^{-1}\otimes 1_N$ belong to 
$\CC$. These two operators will subsequently play an important role.

For $\eta\in  C_0\big([-\infty,+\infty)\big)$, let us now look at the image of 
$\eta(D)\otimes 1_N$ in $\Hrond^\theta$. 
For this, we recall that $\H^1(\R)$ denotes the first Sobolev space on $\R$. 

\begin{Lemma}[Lem.~3.6 in \cite{APRR}]
\begin{enumerate}
\item[$(i)$] For $j\in \{1,\dots,N\}$ the operator
$D^\theta_j:=\big(\V_j^\theta\big)^* D\V^\theta_j$ is self-adjoint on $(\V_j^\theta\big)^* \H^1(\R)$, 
and the following equality holds:
$$
D^\theta_j=2\bigg(1-\Big(\tfrac{X^\theta_j-\lambda_j^\theta}{2}\Big)^2\bigg)\left(-i\tfrac{\d}{\d \lambda}\right) 
+i\Big(\tfrac{X^\theta_j-\lambda^\theta_j}{2}\Big)
$$
where $X^\theta_j$ denotes the operator of multiplication by the variable in $\ltwo(I_j^\theta)$,
\item[$(ii)$] For any $\eta\in  C_0\big([-\infty,+\infty)\big)$ we set
\begin{equation*}
\eta(D^\theta):=\big(\V^\theta\big)^* \big(\eta(D)\otimes 1_N \big)\V^\theta, 
\end{equation*}
and for $\zeta \in \Hrond^\theta$ and for $\lambda\in I^\theta$ one has
$$
\big[\eta(D^\theta) \zeta\big](\lambda) 
= \sum_{\{j\mid\lambda\in I^\theta_j\}} \big[\big(\eta(D^\theta_j)\otimes \P_j^\theta\big)\zeta\big](\lambda).
$$
\end{enumerate}
\end{Lemma}

Let us finally introduce the $C^*$-algebra which plays a key role
in \cite{APRR}.
Unfortunately, this algebra is not sufficient for $\theta = 0$ in the intricate case, for which an alternative construction will be presented subsequently. 
We set
\begin{equation}\label{eq_def_EE}
\EE^\theta:=C^*\Big(\eta(D^\theta)\;\!a+b\mid \eta \in   C_0\big([-\infty,+\infty)\big),  a\in \AA^\theta, b\in \K\big(\Hrond^\theta\big)\Big)^+
\end{equation}
which is a $C^*$-subalgebra of $\B\big(\Hrond^\theta\big)$
containing the ideal $\K\big(\Hrond^\theta\big)$ of compact operators on $\Hrond^\theta$.
Here, the exponent $+$ means that $\C$ times the identity of $\B\big(\Hrond^\theta\big)$ have been added to the algebra, turning it into a unital $C^*$-algebra.
Our main interest in this $C^*$-algebra is that it contains the wave operators
introduced in \eqref{eq_WO}, as shown in \cite[Prop.~3.7]{APRR}. 
In order to understand this statement, more investigations on 
these operators are necessary. We recall the necessary results in the next section.

\subsection{Formulas for the wave operators}

So far, all values of $\theta$ have been treated similarly. 
However, there are two main differences between the case $\theta \in \{0,\pi\}$ and the 
generic case $\theta\in (0,2\pi)\setminus \{\pi\}$. Firstly, the multiplicity of the eigenvalues
$\lambda^\theta_j$ is no more always $1$, and secondly the explicit formulas for the wave operators are, in a very special case for $\theta=0$, slightly more involved. 
For now, we shall disregard this intricate case and come back to it later.
However, in order to define it, we need additional notations.
The material of this section is borrowed from \cite[Sec.~3--5]{NRT}.

We firstly decompose the diagonal matrix $\diag(v)$, with entries $v(1),\dots,v(N)$ on the diagonal, as a product
$$
\diag(v)=\u\;\!\v^2,
$$
where $\v:=|\diag(v)|^{1/2}$ and $\u:=\sgn\big(\diag(v)\big)$ is the diagonal matrix
with components
$$
\u_{j,j}
=\sgn\big(\diag(v)\big)_{j,j}=
\begin{cases}
+1 &\hbox{if $v(j)\ge0$}\\
-1 &\hbox{if $v(j)<0$,}
\end{cases}
\quad j\in\{1,\dots,N\}.
$$
We also introduce the bounded operator $G:\h\to\C^N$ defined by
$G=\v\gamma_0$ with 
$\gamma_0:\h\to\C^N$ given by
$$
(\gamma_0\f)_j:=\int_0^\pi \f_j(\omega)\tfrac{\d\omega}\pi,
\quad \f\in\h,~j\in\{1,\dots,N\}.
$$
Its adjoint  $G^*:\C^N\to\h$ satisfies
$(G^*\xi)_j(\omega):=|v(j)|^{1/2}\xi_j$ for 
$\xi\in\C^N$, $j\in\{1,\dots,N\}$, and $\omega\in[0,\pi)$.
For the resolvent, we set $R_0^\theta(z):=(H_0^\theta-z)^{-1}$ for $z\in \C\setminus \R$.

From the stationary approach of scattering theory, one readily infers that 
for suitable $\f,\g\in \h$
$$
\big\langle\big(W_-^\theta-1\big)\f,\g\big\rangle_\h
=-\int_{\sigma(H_0^\theta)}\lim_{\varepsilon\searrow0}\big\langle G^*
M^\theta(\lambda+i\varepsilon)G\delta_\varepsilon\big(H_0^\theta-\lambda\big)\f,
R_0^\theta(\lambda-i\varepsilon)\g\big\rangle_\h\;\!\d\lambda,
$$
with
$$
M^\theta(z):=\big(\u+GR_0^\theta(z)G^*\big)^{-1},\quad z\in\C\setminus\R
$$
and
$
\delta_\varepsilon\big(H_0^\theta-\lambda\big)
:=\frac{\pi^{-1}\varepsilon}{(H_0^\theta-\lambda)^2+\varepsilon^2}
$
with $\varepsilon>0$ and $\lambda\in\R$.
In order to derive an expression for the operator
$\big(W_-^\theta-1\big)$ in the spectral representation of $H_0^\theta$ we define the set
$$
\textstyle\Drond^\theta
:=\left\{\zeta\in\Hrond^\theta\mid\zeta=\sum_{j=1}^N\zeta_j,
~\zeta_j\in C^\infty_{\rm c}\big(I^\theta_j\setminus
\big(\T^\theta\cup\sigma_{\rm p}(H^\theta)\big);\P^\theta_j\C^N\big)\right\},
$$
which is dense in $\Hrond^\theta$ because $\T^\theta$ is countable and
$\sigma_{\rm p}(H^\theta)$ is closed and of Lebesgue measure $0$. 
Then, it is shown in the mentioned reference that for  
$\zeta,\zeta'\in\Drond^\theta$ one has
\begin{align}
&\big\langle\F^\theta\big(W_-^\theta-1\big)(\F^{\theta })^*\zeta',
\zeta\big\rangle_{\Hrond^\theta}\nonumber\\
&=-\frac{1}{\sqrt{\pi}}\sum_{j=1}^N\int_{I^\theta_j}\lim_{\varepsilon\searrow0}
\left\langle\v M^\theta(\lambda+i\varepsilon)\v\;\!\gamma_0(\F^{\theta})^*
\delta_\varepsilon\big(X^\theta-\lambda\big)\zeta',
\int_{I_j^\theta}\frac{\beta_j^\theta(\mu)^{-1}}{\mu-\lambda+i\varepsilon}\;\!
\zeta_j(\mu)\;\!\d\mu\right\rangle_{\C^N}\d\lambda\label{eq_leading}\\
&-\frac{1}{\sqrt{\pi}}\sum_{j=1}^N\int_{\sigma(H^\theta_0)\setminus I^\theta_j}
\lim_{\varepsilon\searrow0}\left\langle\v M^\theta(\lambda+i\varepsilon)\v\;\!\gamma_0
(\F^{\theta})^*\delta_\varepsilon\big(X^\theta-\lambda\big)\zeta',
\int_{I_j^\theta}\frac{\beta_j^\theta(\mu)^{-1}}{\mu-\lambda+i\varepsilon}\;\!
\zeta_j(\mu)\;\!\d\mu\right\rangle_{\C^N}\d\lambda \label{eq_remainder}
\end{align}
with the positive numbers
$$
\beta_j^\theta(z):=\big|\big(z-\lambda_j^\theta\big)^2-4\big|^{1/4},
\quad j\in\{1,\dots,N\},~z\in\C.
$$

Based on these expressions, one of the main results of \cite{NRT} has been to provide
a rather explicit formula for $W_-^\theta$. Before exhibiting them,
we still need to introduce the precise definition of the intricate case.
As emphasized in \cite[Eq.~1.16]{NRT} this situation takes place only for very specific potentials.
This can be easily observed from the following definition, since $\big(\xi^0_{N}\big)_k=1$ for any $k\in \{1, \dots,N\}$,
while $\big(\xi^0_{N/2}\big)_k=(-1)^k$.

\begin{Definition}\label{def_intricate}
The \emph{intricate case} stands for the case where $\theta=0$,
$N\in2\N$, and $\v\xi_N^0$ and $\v\xi^0_{N/2}$ are linearly dependent.
\end{Definition}

We now use again the unitary transformation introduced in \eqref{eq:Vjt}
and \eqref{eq:Vjt*}, and the operators $X$ and $D$ introduced right after.
By using these notations, it has been shown in \cite[Sec.~4.4]{NRT} that for any
$\theta\in[0,2\pi)$ the equality
\begin{align}\label{eq_formula}
\nonumber &\F^\theta\big(W_-^\theta-1\big)(\F^\theta)^* \\
& =\tfrac12(\V^{\theta})^*\big\{\big(1-\tanh(\pi D)-i\cosh(\pi D)^{-1}\tanh(X)\big)
\otimes1_N\big\}\V^\theta\big(S^\theta(X^\theta)-1\big)+K^\theta
\end{align}
holds, with $K^\theta\in\K(\Hrond^\theta)$ in the non-intricate cases, and
$K^0\in\B(\Hrond^0)$ in the intricate case.
Observe that this statement leads directly to the next statement:

\begin{Proposition}\label{prop:affiliation}
For any $\theta \in [0,2\pi)$, but not in the intricate case,
the operator
$\F^\theta W_-^\theta (\F^\theta)^*$ belongs to $\EE^\theta$.
\end{Proposition}

Note that this statement does not hold in the intricate case, since the remainder
term in \eqref{eq_formula} is not compact in this special case. The intricate
case will be treated in Section \ref{sec_Degenerate}.

\section{Quotient algebra and topological outcomes}\label{sec:C*}
\setcounter{equation}{0}

The $C^*$-algebra $\EE^\theta$ was introduced in \eqref{eq_def_EE}, 
and the affiliation of the wave
operator to this algebra was precisely stated in Proposition \ref{prop:affiliation}.
The computation of the quotient of this algebra by the ideal of compact
operators was performed in \cite{APRR} only for $\theta \in (0,\pi)$.
The result remains valid for $\theta \in (\pi,2\pi)$ but is not valid for the
special cases $\theta \in \{0,\pi\}$. Note that this difference is due to the multiplicity of some eigenvalues $\lambda_j^\theta$, as explained below. 
In this section, we extend the results of \cite{APRR} to arbitrary $\theta$.

\subsection{Quotient algebra}

Our first task is to compute the quotient of the algebra $\EE^\theta$ by the set of compact operators. This construction was exhibited in \cite[Sec.~3]{APRR} for $\theta \in (0,\pi)$, and our current aim is to extend it to arbitrary $\theta \in [0,2\pi)$. The main difficulty in the extension is that for $\theta\in \{0,\pi\}$, some eigenvalues of $A^\theta$, as given in \eqref{eq_eigenvalues},
are of multiplicity $2$, while they are always of multiplicity $1$ for $\theta\not \in \{0,\pi\}$.
Namely, for $\theta \not \in \{0,\pi\}$, the eigenvalues $\lambda^\theta_j$ can be uniquely ordered by increasing order. For example, for $\theta\in (0,\pi)$ and for $N$ even, we have
$$
\lambda^\theta_{\frac{N}{2}}<\lambda^\theta_{\frac{N}{2}-1}<\lambda^\theta_{\frac{N}{2}+1}<\lambda^\theta_{\frac{N}{2}-2}<\lambda^\theta_{\frac{N}{2}+2}<\ldots < \lambda^\theta_1<\lambda^\theta_{N-1}<\lambda^\theta_{N},
$$
while for $N$ odd we have
$$
\lambda^\theta_{\frac{N-1}{2}}<\lambda^\theta_{\frac{N+1}{2}}<\lambda^\theta_{\frac{N-1}{2}-1}<\lambda^\theta_{\frac{N+1}{2}+1}<\lambda^\theta_{\frac{N-1}{2}-2}<\ldots<\lambda^\theta_1<\lambda^\theta_{N-1}<\lambda^\theta_{N}.
$$
A similar unique ordering holds for $\theta \in (\pi,2\pi)$. On the other hand, for $\theta=0$ and $N$ even we have
$$
-2=\lambda^0_{\frac{N}{2}}<\lambda^0_{\frac{N}{2}+1}=\lambda^0_{\frac{N}{2}-1}<\lambda^0_{\frac{N}{2}+2}=\lambda^0_{\frac{N}{2}-2}<\ldots<\lambda_1^0=\lambda^0_{N-1}<\lambda^0_{N}=2,
$$
while for $N$ odd we have
$$
-2<\lambda^0_{\frac{N-1}{2}}=\lambda^0_{\frac{N+1}{2}}<\lambda^0_{\frac{N-3}{2}}=\lambda^0_{\frac{N+3}{2}}<\ldots<\lambda_1^0=\lambda^0_{N-1}<\lambda^0_{N}=2.
$$
Accordingly, for $\theta=\pi$ and $N$ even we have
$$
-2<\lambda^\pi_{\frac{N}{2}}=\lambda^\pi_{\frac{N}{2}-1}<\lambda^\pi_{\frac{N}{2}+1}=\lambda^\pi_{\frac{N}{2}-2}<\ldots<\lambda_1^\pi=\lambda^\pi_{N-2}<\lambda^\pi_{N}=\lambda^\pi_{N-1}<2,
$$
while for $N$ odd we have
$$
-2=\lambda^\pi_{\frac{N-1}{2}}<\lambda^\pi_{\frac{N-1}{2}-1}=\lambda^\pi_{\frac{N-1}{2}+1}<\ldots<\lambda^\pi_{N-2}=\lambda_1^\pi<\lambda^\pi_{N-1}=\lambda^\pi_{N}<2.
$$
Thus, if we denote by $M$ the number of distinct eigenvalues in the set
$\{\lambda^\theta_j\}_{j=1}^N$, then one has
$M=N$ for $\theta \not \in \{0,\pi\}$, while $M= \frac{N}{2}+1$ for $\theta=0$ and $N$ even, 
$M=\frac{N+1}{2}$ for $\theta \in \{0,\pi\}$ and $N$ odd, and $M= \frac{N}{2}$ for $\theta=\pi$ and $N$ even. 

We are now interested in computing the quotient algebra $\QQ^\theta:=\EE^\theta / \K\big(\Hrond^\theta\big)$. 
As a consequence of \cite[Lem.~3.9]{APRR} on the compactness of commutators $[\eta(D^\theta),a]$, we can focus on elements of the form $\eta(D^\theta)a$
with $\eta \in C_0\big([-\infty,+\infty)\big)$ and $a\in \AA^\theta$ 
(introduced in Definition \ref{def_A_theta}).
The starting point is the decomposition
\begin{equation}\label{eq:sum}
\eta(D^\theta)\;\!a
=  \sum_{j,j'}
\big(\eta(D^\theta_{j})\;\!\a_{j,j'}\big)\otimes \;\!\tfrac{1}{N}|\xi^\theta_j\rangle \langle \xi^\theta_{j'}|
\end{equation}
which follows from \eqref{eq_decc}.
We can then deal with operators of the form
$\eta(D^\theta_{j})\,\!\a_{j,j'}: \ltwo(I^\theta_{j'})\to \ltwo(I^\theta_j)$.
However, since $\supp \a_{j,j'}\subset \overline{I_j^\theta \cap I_{j'}^\theta}$, it is 
simpler to look at this operator in $\ltwo(I_j^\theta)$.
Thus, if we set $q_{j}:\B\big(\ltwo(I_{j}^\theta)\big)\to \B\big(\ltwo(I_{j}^\theta) \big)/ \K\big(\ltwo(I_{j}^\theta) \big)$ then the image of  $\eta(D^\theta_{j})\,\!\a_{j,j'}$ through
$q_{j}$ falls into two distinct situations.
\begin{enumerate}
\item[$(i)$] If $I_j^\theta \neq I^\theta_{j'}$, then 
\begin{equation}\label{eq:q1}
q_{j}\big(\eta(D^\theta_{j})\;\!\a_{j,j'}\big) 
=  \eta(-\infty)\;\!\a_{j,j'} \in   C_0\big(I^\theta_{j}\cap I^\theta_{j'}\big),
\end{equation}
\item[$(ii)$]   If $I_j^\theta=I_{j'}^\theta$ (meaning $\lambda^\theta_j=\lambda^\theta_{j'})$, then 
\begin{align}\label{eq:q2}
\begin{split}
q_{j}\big(\eta(D^\theta_{j})\;\!\a_{j,j'}\big)
& = \Big( \eta\;\!\a_{j,j'}(\lambda_{j}^\theta-2),\ \eta(-\infty)\;\!\a_{j,j'},\ 
 \eta\;\!\a_{j,j'}(\lambda_{j}^\theta+2) \Big) \\ 
& \qquad \in
C_0\big((+\infty,-\infty]\big) \oplus C\big(\overline {I^\theta_{j}}\big) \oplus 
C_0\big([-\infty,+\infty)\big).
\end{split}
\end{align}
\end{enumerate}
These statement can be obtained as in the proof of \cite[Lem.~3.8]{APRR} by looking
at these operators in $\ltwo(\R)$ through the conjugation by $\V_{j}^\theta$. Then, one ends up with operators of the form $\eta(D)\varphi(X)$ for $\eta \in C_0\big([-\infty,+\infty)\big)$
and for $\varphi\in C_0(\R)$ in the first case, and $\varphi \in C\big([-\infty,+\infty]\big)$
in the second case. The image of such operators by the quotient map (defined by the compact operators) have been extensively studied in \cite[Sec.~4.4]{Ri}, from which we infer the results presented above. 

\begin{Remark}\label{rem_identification}
Observe that in the first component of \eqref{eq:q2}, the interval $[-\infty,+\infty)$ has been oriented in the reverse direction. The reason is that the function
introduced in \eqref{eq:q2} can be seen as a continuous function on the union of the three intervals
$$
(+\infty,-\infty]\cup \overline {I^\theta_{j}} \cup [-\infty,+\infty)
$$
once their endpoints are correctly identified. This observation and this trick will be used several times in the sequel.
Note also that \eqref{eq:q1} could be expressed as \eqref{eq:q2} by considering
the triple $\big(0, \;\! \eta(-\infty)\a_{j,j'},\;\! 0\big)$.
\end{Remark}

In the next statement we collect the results obtained so far. 
However, in order to provide a unified statement for all $\theta\in [0,2\pi)$ and for arbitrary $N$,  some notations have to be slightly updated.
More precisely, recall that for $j\in\{1,\dots,N\}$ one has
$\lambda_j^\theta:=2\cos\big(\frac{\theta+2\pi\;\!j}N\big)$.
For fixed $\theta$, let us now denote by $\tl_1<\tl_2<\ldots<\tl_{M}$
these $M$ distinct values.
As already mentioned, their multiplicities are generically $1$, but can be $2$ when $\theta =0$ or $\theta =\pi$. 
Accordingly, we define the orthogonal projection $\tP_k$ (based on the projections $\P^\theta_j$ corresponding to the eigenvalue(s) $\tilde\lambda^\theta_k=\lambda_j^\theta$) which is
one dimensional if the corresponding eigenvalue is of multiplicity $1$, and two dimensional when $\lambda_j^\theta$ is of multiplicity $2$.
Finally for $k\in \{1, \dots, M\}$ we set
$$
\tH_k:=\sp\big\{\tP_i\C^N\mid i\leq k\big\}.
$$

In order to understand the next statement, 
let us mention that a very specific feature appears in the special case $\theta=0$
and $N$ even. Indeed, in this case $-2$ and $2$ are eigenvalues of the matrix $A^0$, 
namely $\lambda^0_{N/2}=-2$ and $\lambda^0_N=2$.
As a consequence, the closure of the two intervals $I^0_{N/2}=(-4,0)$ and $I^0_N=(0,4)$
intersects at $0$, and this creates a very special threshold in the spectrum of $H^0_0$:
a doubly degenerate threshold. It corresponds simultaneously to the closure of one scattering channel, and to the opening of a new one.
This specificity appears in the computation of the quotient algebra.

\begin{Proposition}\label{prop_quotient}
1) For $\theta \neq 0$ or $N\not \in 2\N$, the quotient algebra $\QQ^\theta:=\EE^\theta / \K\big(\Hrond^\theta\big)$ has the shape
of an upside down comb, and more precisely:
\begin{equation}\label{eq:comb}
\QQ^\theta \subset C\left(\Big(\bigoplus_{k=1}^{M-1}  \downarrow_k\oplus \rightarrow_k\Big)
\oplus \Big(\downarrow_M \oplus \rightarrow_M \oplus \uparrow^1 \Big) \oplus  \Big( \bigoplus_{k=2}^{M}  \rightarrow^k\oplus \uparrow^k\Big);\B(\C^N)\right),
\end{equation}
with
\begin{align*}
\downarrow_k & := [+\infty,-\infty] \qquad \forall k \in \{1, \ldots,M\}\\
\rightarrow_k & := [\tl_k-2, \tl_{k+1}-2]  \qquad \forall k \in \{1, \ldots,M-1\} \\
\rightarrow_M & :=  [\tl_M-2, \tl_{1}+2] \\
\rightarrow^k & :=  [\tl_{k-1}+2, \tl_{k}+2]  \qquad \forall k \in \{2, \ldots,M\} \\
\uparrow^k & := [-\infty,+\infty]  \qquad \forall k \in \{1, \ldots,M\}.
\end{align*}
Moreover, if $\phi_*$ denotes the restriction to the edge $*$ of any $\phi \in \QQ^\theta$, 
then these restrictions satisfy the conditions:
\begin{align*}
\phi_{\rightarrow_k}  & \in C\big( [\tl_k-2, \tl_{k+1}-2];\B(\tH_k)\big) \ \ \forall k\in \{1, \ldots, M-1\}\\
\phi_{\rightarrow_M}  & \in C\big( [\tl_M-2, \tl_{1}+2]; \B(\C^N)\big) \\
\phi_{\rightarrow^k}  & \in C\big( [\tl_{k-1}+2, \tl_{k}+2]; \B((\tH_{k-1})^\bot)\big) \ \ \forall k\in \{2, \ldots, M\},
\end{align*}
together with
\begin{align*}
&\phi_{\downarrow_k} \in C\big([+\infty,-\infty]; \B(\tP_k \C^N) \big)  \\
&\phi_{\uparrow^{k}} \in C\big([-\infty,+\infty];\B(\tP_k\C^N)\big),
\end{align*}
for all $k\in \{1, \ldots, M\}$.
In addition,  the following continuity properties hold:
\begin{align}
\label{eq:cont1}\phi_{\rightarrow_1}(\tl_1-2) &= \phi_{\downarrow_1}(-\infty) \\
\label{eq:cont2}\phi_{\rightarrow_k}(\tl_k-2) &= \phi_{\rightarrow_{k-1}}( \tl_k-2)
\oplus \phi_{\downarrow_k}(-\infty) \ \ \forall k\in \{2, \ldots, M\} \\
\label{eq:cont3}\phi_{\rightarrow_M}(\tl_1+2)  &=   \phi_{\uparrow^{1}}(-\infty) \oplus \phi_{\rightarrow^{2}}(\tl_1+2) \\
\label{eq:cont4}\phi_{\rightarrow^k}(\tl_k+2)  &=  \phi_{\uparrow^{k}}(-\infty) \oplus \phi_{\rightarrow^{k+1}}(\tl_k+2) \ \ \forall k\in \{2, \ldots, M-1\} \\
\label{eq:cont5}\phi_{\rightarrow^M}(\tl_M+2)  &= \phi_{\uparrow^{M}}(-\infty),
\end{align}
and there exists $c\in \C$ such that for all $k\in \{1, \dots,M\}$,
\begin{equation}\label{eq:unit}
\phi_{\downarrow_k}(+\infty)= c\;\! \tP_k = \phi_{\uparrow^{k}}(+\infty).
\end{equation}

2) In the special case $\theta=0$ and $N\in 2\N$,  the centered part of \eqref{eq:comb} has to be modified into
$\big( \uparrow^1 \oplus  \downarrow_M \big)$, the function $\phi_{\rightarrow_M}$ does not exist, and the following continuity properties hold:
\begin{align}
\label{eq:cont1b}\phi_{\rightarrow_1}(-4) &= \phi_{\downarrow_1}(-\infty), \\
\label{eq:cont2b}\phi_{\rightarrow_k}( \tilde \lambda^0_k-2) &= \phi_{\rightarrow_{k-1}}(  \tilde \lambda^0_k-2)
\oplus \phi_{\downarrow_k}(-\infty) \ \ \forall k\in \{2, \ldots, M-1\} \\
\label{eq:cont3b}\phi_{\rightarrow_{M-1}}(0)  &=   \phi_{\uparrow^{1}}(-\infty) \oplus \phi_{\rightarrow^{2}}(0)\big|_{\sp\{\tilde \P^0_i\C^N\mid 1<i< M\}} \\
\label{eq:cont4b}\phi_{\rightarrow^{2}}(0)  &=  \phi_{\rightarrow_{M-1}}(0)\big|_{\sp\{\tilde \P^0_i\C^N\mid 1<i< M\}}
 \oplus  \phi_{\downarrow_{M}}(-\infty)\\
\label{eq:cont5b}\phi_{\rightarrow^k}( \tilde \lambda^0_k+2)  &=  \phi_{\uparrow^{k}}(-\infty) \oplus \phi_{\rightarrow^{k+1}}( \tilde \lambda^0_k+2) \ \ \forall k\in \{2, \ldots, M-1\} \\
\label{eq:cont6b}\phi_{\rightarrow^M}(4)  &= \phi_{\uparrow^{M}}(-\infty).
\end{align}
\end{Proposition}

A representation for the support of the quotient algebra $\QQ^\theta$ is provided in Figure \ref{fig_Rt}. 
In the previous description of the quotient algebra, note that the condition \eqref{eq:unit} is related to the addition of the unit to $\EE^\theta$. Indeed, if no unit is added to $\EE^\theta$, then one has $c=0$. 
In the proof we shall use the notation $q^\theta$ for the quotient map 
$$
q^\theta: \EE^\theta \to \QQ^\theta\equiv \EE^\theta/\K\big(\Hrond^\theta\big).
$$
\begin{figure}
    \centering
    \includegraphics[width=15cm]{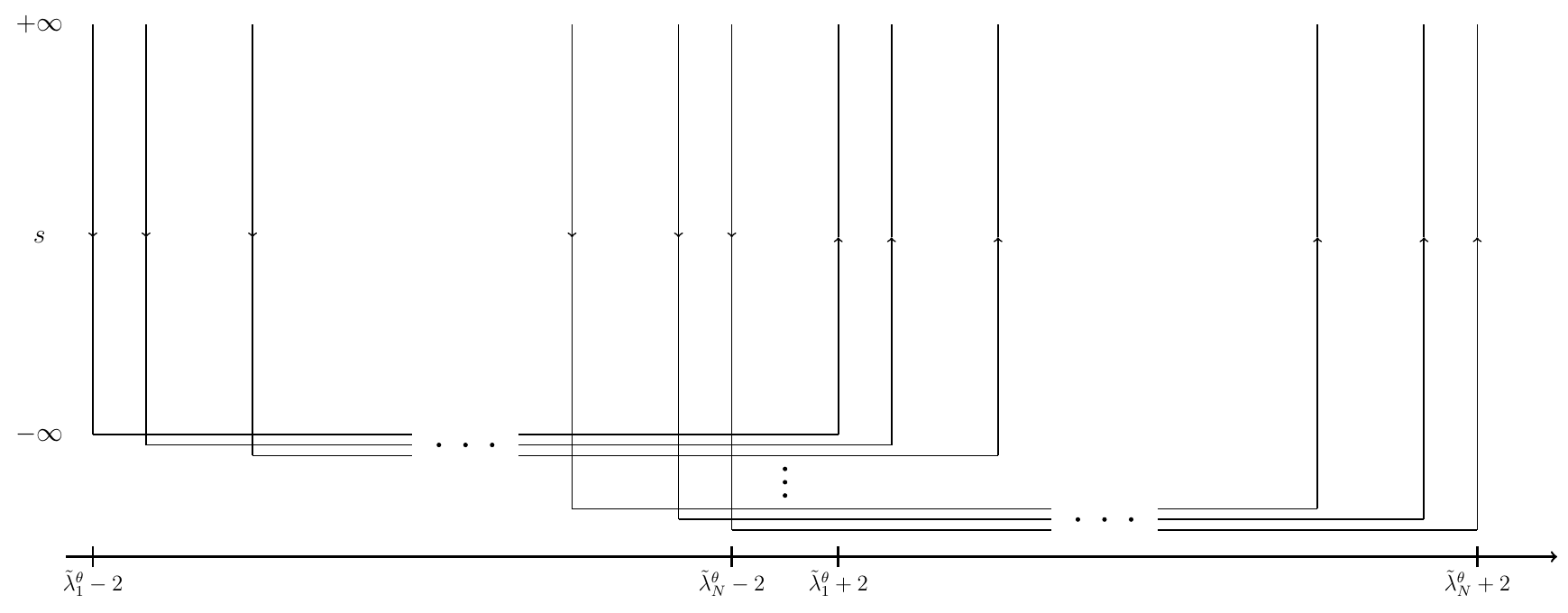}
    \caption{A representation of the quotient algebra $\QQ^\theta$ for $\theta \in (0,\pi)$.}
    \label{fig_Rt}
\end{figure}

\begin{proof}
The proof consists in looking carefully at the expressions provided in \eqref{eq:q1} and \eqref{eq:q2}, and in keeping track of the matricial form of the matrix-valued function $a$.
Let us consider a generic element of $\EE^\theta$ given by
$\eta(D^\theta)a + c1$ for $\eta\in C_0\big([-\infty,+\infty)\big)$, $a\in \AA^\theta$,
and $c\in \C$. Note that we already disregarded a compact term, since we
are going to perform the quotient by ideal of compact operators.
By taking \eqref{eq:sum} into account, and the specific form of $\Hrond^\theta$,
one has
\begin{equation*}
\eta(D^\theta)a + c1
=  \sum_{j,j'}
\big(\eta(D^\theta_{j})\;\!\a_{j,j'}+c \delta_{j,j'}\big)\otimes 
\;\!\tfrac{1}{N}|\xi^\theta_j\rangle \langle \xi^\theta_{j'}|
\end{equation*}
with $\delta_{j,j'}$ the Kronecker delta function.
According to \eqref{eq:q1} and \eqref{eq:q2} it then follows that
\begin{align*}
&q^\theta\big(\eta(D^\theta)a + c1\big) \\
& = \sum_{\{j,j'\mid I_{j}^\theta\neq I_{j'}^\theta\}}
\Big(0,\eta(-\infty)\;\! \a_{j,j'},0\Big)\otimes 
\;\!\tfrac{1}{N}| \xi^\theta_j\rangle \langle  \xi^\theta_{j'}| \\
&+ \sum_{\{j,j'\mid I_{j}^\theta= I_{j'}^\theta\}}
\Big(\eta\;\! \a_{j,j'}( \lambda_{j}^\theta-2)+c \delta_{j,j'},\eta(-\infty)\;\! \a_{j,j'}+c \delta_{j,j'}, 
\eta\;\! \a_{j,j'}( \lambda_{j}^\theta+2)+c \delta_{j,j'} \Big)\otimes 
\;\!\tfrac{1}{N}| \xi^\theta_j\rangle \langle  \xi^\theta_{j'}|. 
\end{align*}
For the second summation, recall that the multiplicity of
an eigenvalue is at most two, which means that 
the condition $ I_{j}^\theta= I_{j'}^\theta$ holds for at most
two different $j,j'$. Note that this condition is also equivalent to $\lambda_{j}^\theta= \lambda_{j'}^\theta=:\tilde \lambda_k^\theta$ for some $k\in \{1,\dots,M\}$.
Let us set $\tilde \P_k^\theta$ for the corresponding one or two dimensional projection. Then the second summation can be rewritten as
$$
\sum_{k\in \{1, \dots,M\}}
\Big( \eta\;\! \tilde a_{k}( \tilde\lambda_{k}^\theta-2)+c\tilde \P_k^\theta,\ \eta(-\infty)\;\! \tilde a_{k}+c\tilde \P_k^\theta, 
\ \eta\;\! \tilde a_{k}( \tilde \lambda_{k}^\theta+2)+c\tilde \P_k^\theta \Big)
$$
where $\tilde a_{k}$ is defined by 
$$
\tilde a_{k}:=\sum_{\{j,j'\mid \lambda_{j}^\theta= \lambda_{j'}^\theta=:\tilde \lambda_k^\theta\}}\a_{j,j'}\otimes \tfrac{1}{N}| \xi^\theta_j\rangle \langle  \xi^\theta_{j'}|.
$$
Clearly, $\tilde a_k$ also acts in the subspace defined by $\tilde \P_k^\theta$.

1) In order to fully understand the previous expression, it is necessary to remember
the special structure of the underlying Hilbert space.
When $\theta \neq 0$ or $N\not \in 2\N$ one has 
$\Hrond^\theta:=\int_{I^\theta}^\oplus\Hrond^\theta(\lambda)\;\!\d\lambda$
with 
\begin{align*}
\Hrond^\theta(\lambda)
& =
\begin{cases}  \tilde \Hrond^\theta_k  & \hbox{ if } \tilde\lambda_k^\theta-2\leq \lambda< \tilde\lambda_{k+1}^\theta-2,  \qquad k\in \{1, \dots, M-1\}\\
\C^N & \hbox{ if } \tilde \lambda^\theta_M-2 \leq \lambda \leq \tilde \lambda^\theta_1+2,\\
(\tilde \Hrond^\theta_k)^\bot  & \hbox{ if } \tilde\lambda_{k}^\theta +2<\lambda\leq \tilde\lambda_{k+1}^\theta+2,  \qquad k\in \{1, \dots, M-1\}. 
\end{cases}
\end{align*}
Note that compared with the original definition of $\Hrond^\theta(\lambda)$ we have changed the fiber at a finite number of points, which does 
not impact the direct integral, but simplify our argument subsequently.
Thus, the changes of dimension of the fibers take place at all $\tilde \lambda_k^\theta-2$ 
and $\tilde \lambda_k^\theta+2$, for $k\in \{1, \dots,M\}$.
By taking this into account, the interval $I^\theta$ has to be divided into $2M-1$ subintervals, 
firstly of the form $[\tilde \lambda_k^\theta-2, \tilde \lambda_{k+1}^\theta-2)$ for
$k\in \{1,\dots,M-1\}$, then
the special interval $[\tilde \lambda_M^\theta-2,\tilde \lambda_1^\theta+2]$, and finally
the intervals $(\tilde \lambda_{k-1}^\theta+2,\tilde \lambda_{k}^\theta+2]$ for $k\in \{2, \dots, M\}$.
By using this partition of $I^\theta$, one gets
\begin{align*}
&q^\theta\big(\eta(D^\theta)a + c1\big) \\
&=\sum_{k=1}^{M-1}
\Big(\big(\eta\;\!\tilde a_{k}(\tilde \lambda_{k}^\theta-2)+c\tilde \P_k^\theta\big),\  \chi_{[\tilde \lambda_k^\theta-2, \tilde \lambda_{k+1}^\theta-2)}\big(
\eta(-\infty)\;\!a+c1_{\tilde{\Hrond}^\theta_k}\big)\Big) \\
& \quad +
\Big(\big(\eta\;\!\tilde a_{M}(\tilde \lambda_{M}^\theta-2)+c\tilde \P_M^\theta\big),
\  \chi_{[\tilde \lambda_M^\theta-2,\tilde \lambda_1^\theta+2]}\big(\eta(-\infty)\;\!a+c1_{\C^N}\big),\  
\big(\eta\;\!\tilde a_{1}(\tilde \lambda_{1}^\theta+2)+c\tilde \P_1^\theta\big) \Big) \\
& \quad + \sum_{k=2}^{M} 
\Big(\chi_{(\tilde \lambda_{k-1}^\theta+2, \tilde \lambda_{k}^\theta+2]}\big(\eta(-\infty)\;\!a+c1_{(\tilde{\Hrond}^\theta_{k-1})^\bot}\big), \ 
\big(\eta\;\!\tilde a_{k}(\tilde \lambda_{k}^\theta+2)+c\tilde \P_k^\theta\big) \Big).
\end{align*}
The description obtained so far leads directly to the structure of \eqref{eq:comb}.

The properties stated in equations \eqref{eq:cont1} to \eqref{eq:cont5} follow from the continuity of $a\in \AA^\theta$
and from its properties at thresholds. 
The final property \eqref{eq:unit} comes from the unit and the fact that 
$\eta$ vanishes at $+\infty$.

2) When $\theta = 0$ and  $N \in 2\N$, the special interval
$[\tilde \lambda_M^\theta-2,\tilde \lambda_1^\theta+2]$ mentioned above does not exist. In this case one has
$\Hrond^0:=\int_{[-4,4]}^\oplus\Hrond^0(\lambda)\;\!\d\lambda$
with 
\begin{align*}
\Hrond^0(\lambda)
& =
\begin{cases}  \tilde \Hrond^0_k  & \hbox{ if } \tilde\lambda_k^0-2\leq \lambda< \tilde\lambda_{k+1}^0-2,  \qquad k\in \{1, \dots, M-2\}\\
\tilde \Hrond^0_{M-1}  = \big[\tilde \lambda^0_{M-1}-2,0\big] \\
(\tilde \Hrond^0_{1})^\bot  = \big[0,\tilde \lambda^0_{2}+2\big] \\
(\tilde \Hrond^0_k)^\bot  & \hbox{ if } \tilde\lambda_{k}^0 +2<\lambda\leq \tilde\lambda_{k+1}^0+2,  \qquad k\in \{2, \dots, M-1\}. 
\end{cases}
\end{align*}
One then gets
\begin{align*}
&q^0\big(\eta(D^0)a + c1\big) \\
&=\sum_{k=1}^{M-2}
\Big(\big(\eta\;\!\tilde a_{k}(\tilde \lambda_{k}^0-2)+c\tilde \P_k^0\big),\  \chi_{[\tilde \lambda_k^0-2, \tilde \lambda_{k+1}^0-2)}\big(
\eta(-\infty)\;\!a+c1_{\tilde{\Hrond}^0_k}\big)\Big) \\
& \quad +
\Big(\big(\eta\;\!\tilde a_{M-1}(\tilde \lambda_{M-1}^0-2)+c\tilde \P_{M-1}^0\big),
\  \chi_{[\tilde \lambda_{M-1}^0-2,0]}\big(\eta(-\infty)\;\!a+c1_{\tilde{\Hrond}^0_{M-1}}\big),\  
\big(\eta\;\!\tilde a_{1}(0)+c\tilde \P_1^0\big) \Big) \\
& \quad +
\Big(\big(\eta\;\!\tilde a_{M}(0)+c\tilde \P_M^0\big),
\  \chi_{[0,\tilde \lambda_2^0+2]}\big(\eta(-\infty)\;\!a+c1_{(\tilde{\Hrond}^0_{1})^\bot}\big),\  
\big(\eta\;\!\tilde a_{2}(\tilde \lambda_{2}^0+2)+c\tilde \P_2^0\big) \Big) \\
& \quad + \sum_{k=3}^{M} 
\Big(\chi_{(\tilde \lambda_{k-1}^0+2, \tilde \lambda_{k}^0+2]}\big(\eta(-\infty)\;\!a+c1_{(\tilde{\Hrond}^0_{k-1})^\bot}\big), \ 
\big(\eta\;\!\tilde a_{k}(\tilde \lambda_{k}^0+2)+c\tilde \P_k^0\big) \Big).
\end{align*}
This description leads directly to the continuity properties \eqref{eq:cont1b} to \eqref{eq:cont6b} in the special case $\theta = 0$ and  $N \in 2\N$.
\end{proof}

Since $\F^\theta W_-^\theta(\F^\theta)^*\in \EE^\theta$, by Proposition \ref{prop:affiliation},
one can look at the image of this operator in the quotient algebra.
The following statement contains a description of this image, using the notations introduced
in Proposition \ref{prop_quotient}. 
For this description, we define the two functions $\eta_\pm:\R\to \C$ for $s \in \R$ by 
\begin{equation}\label{eq_def_eta}
\eta_\pm(s):=\tanh(\pi s)\pm i\sech(\pi s).
\end{equation}
Finally, based on the projections $\{\tilde \P^\theta_k\}_{k=1}^M$
introduced before Proposition \ref{prop_quotient}, we define the channel
scattering matrix $\tilde S^\theta_{k}(\lambda):=\tilde \P_k^\theta S^\theta(\lambda)\tilde \P_k^\theta$. 

\begin{Lemma}\label{lem:restrictions}
For any $\theta \in [0,2\pi)$, but not in the intricate case, let 
$$
\phi:=q^\theta \big(\F^\theta W_-^\theta(\F^\theta)^*\big)
$$
denote the image of $\F^\theta W_-^\theta(\F^\theta)^*$ in the quotient algebra. 
Then, when $\theta \neq 0$ or $N\not \in 2\N$ the restrictions of $\phi$ on the various parts of $\QQ^\theta$ are given by:
\begin{align*}
& \hbox{ for }  k\in \{1, \ldots, M-1\} \hbox{ and } \lambda \in [\tl_k-2, \tl_{k+1}-2), \quad & \phi_{\rightarrow_k}(\lambda) = S^\theta(\lambda) \ \in \B(\tH_k)\\
&\hbox{ for } \lambda \in  [\tl_M-2, \tl_{1}+2], \quad & \phi_{\rightarrow_M}(\lambda) 
= S^\theta(\lambda)\  \in \B(\C^N)\\
&\hbox{ for }  k\in \{2, \ldots, M\} \hbox{ and }  \lambda \in  (\tl_{k-1}+2, \tl_{k}+2], \quad & \phi_{\rightarrow^k}(\lambda) = S^\theta(\lambda) \ \in  \B\big((\tH_{k-1})^\bot\big),
\end{align*}
and for $s\in \R$
\begin{align}
\label{eq:res1} &\phi_{\downarrow_k}(s)= 1+ \tfrac12\big(1-\eta_-(s)\big)\big(\tilde S^\theta_{k}(\tilde\lambda^\theta_k-2)-1\big) \ \in \B(\tP_k \C^N) \\
\label{eq:res2} &\phi_{\uparrow^{k}}(s)= 1+ \tfrac12\big(1-\eta_+(s)\big)\big(\tilde S^\theta_{k}(\tilde\lambda^\theta_k+2)-1\big) \ \in\B(\tP_k\C^N).
\end{align}
When $\theta = 0$ and $N \in 2\N$ the expressions are the same except that $\phi_{\rightarrow_M}$ does not exist. 
\end{Lemma}

\begin{proof}
Let us start by looking at the expression for
$\F^\theta \big(W_-^\theta-1\big)(\F^\theta)^*$, as provided in 
\eqref{eq_formula}, and by taking the content of Lemma \ref{lem_tanh}
into account. It then follows that $\F^\theta \big(W_-^\theta-1\big)(\F^\theta)^*$
can be rewritten as
\begin{equation}\label{eq:2terms}
\tfrac12\big(1-\tanh(\pi D^\theta)\big)\big(S^\theta(X^\theta)-1\big)
-i\tfrac{1}{2}\cosh(\pi D^\theta)^{-1} \sum_j
\tfrac{X^\theta-\lambda^\theta_j}{2}\P^\theta_j 
\big(S^\theta(X^\theta)-1\big)
+k^\theta
\end{equation}
with $k^\theta\in \K\big(\Hrond^\theta\big)$.
Thus, one ends up with two generic elements of $\EE^\theta$ which have been carefully
studied in the proof of Proposition \ref{prop_quotient}. The only additional necessary tricky observation is that
$$
\tfrac{\lambda-\lambda^\theta_j}{2}\Big|_{\lambda = \lambda_j^\theta\pm 2} = \pm 1
$$
which explains the appearance of the two functions $\eta_\pm$.
Note also that 
$$
\lim_{s\to -\infty}\tfrac12\big(1-\tanh(\pi s)\big) = 1 \quad \hbox{ and }\quad
\lim_{s\to \infty}\tfrac12\big(1-\tanh(\pi s)\big) = 0, 
$$
while $\lim_{s\to \pm \infty}\cosh(\pi s)^{-1} = 0$.
The rest of the proof is just a special instance of the proof of Proposition \ref{prop_quotient}
for the two main terms exhibited in \eqref{eq:2terms}, since the compact term verifies
$q^\theta(k^\theta)=0$.
\end{proof}

\subsection{Topological Levinson's theorem}\label{sec:Lev}

In this section we provide the topological version of Levinson's theorem, linking
the number of bound states to an expression involving the image of $W^\theta_-$
in the quotient algebra. First of all, we generalize \cite[Lem.~4.1]{APRR} and 
provide a statement about the behavior
of the scattering matrix at thresholds. It shows that the expression obtained for 
the wave operator $W_-^\theta$ is rather rigid and imposes a strict behavior at thresholds.

In the sequel and when $\tilde \lambda_k^\theta$ is an eigenvalue of $A^\theta$
of multiplicity $2$ we shall write ${\tilde{\s}}^\theta_{k}(\tilde\lambda^\theta_k\pm 2)\in M_2(\C)$ for the components of ${\tilde{S}}^\theta_{k}(\tilde\lambda^\theta_k\pm 2)$
in the orthonormal basis defined by the elements of the family $\{\tfrac{1}{\sqrt{N}}\xi^\theta_j\}$ generating $\tilde \P^\theta_k$.

\begin{Lemma}\label{lem:11}
For any $\theta \in [0,2\pi)$, but not in the intricate case,
the matrices $\tilde{S}^\theta_{k}(\tilde\lambda^\theta_k-2)$
are unitary and real.
If  $\tilde \P^\theta_k$ is one dimensional, then ${\tilde{\s}}^\theta_{k}(\tilde\lambda^\theta_k\pm 2)=\pm 1$, while if $\tilde \P^\theta_k$ is two dimensional,
then ${\tilde{\s}}^\theta_{k}(\tilde\lambda^\theta_k\pm 2)$ takes the form 
\begin{equation}\label{eq_form_mat}
\pm \left(\begin{matrix}
1& 0 \\  0 & 1
\end{matrix}
\right)
\qquad \hbox{or}\qquad
\left(\begin{matrix}
a& b \\  \overline{b} & -a
\end{matrix}
\right)
\end{equation}  
with $a\in \R$, $b\in \C$ verifying $a^2+|b|^2=1$.
\end{Lemma}

\begin{proof}
Since $W_-^\theta$ is an isometry and a Fredholm operator, the image
$q^\theta \big(\F^\theta W_-^\theta(\F^\theta)^*\big)$ is a unitary operator, and therefore
its restrictions on all components of $\QQ^\theta$ must be unitary.
The restrictions on $\rightarrow_k$, $\rightarrow_M$, and $\rightarrow^k$ do not impose
any conditions, since the scattering operator is unitary-valued, see  Lemma \ref{lem:restrictions}.
On the other hand, by checking that the restrictions on $\downarrow_k$ and 
on $\uparrow^k$, the following conditions appear.
Starting with \eqref{eq:res1}, the restriction on $\downarrow_k$, by imposing that the operator is unitary-valued, one infers that
$$
\Big(1+ \tfrac12\big(1-\eta_-(s)\big)\big({\tilde{S}}^\theta_{k}(\tilde\lambda^\theta_k-2)-1\big)\Big)
\Big(1+ \tfrac12\big(1-\eta_-(s)\big)\big(\tilde{S}^\theta_{k}(\tilde\lambda^\theta_k-2)-1\big)\Big)^*=1
$$
for all $s\in \R$. Some direct computations lead then to the condition
$\Im\big(\eta_-(s)\big)\Im\big({\tilde{S}}^\theta_{k}(\tilde\lambda^\theta_k-2)\big)=0$
for any $s\in \R$, implying that $\Im\big({\tilde{S}}^\theta_{k}(\tilde\lambda^\theta_k-2)\big)=0$.
Since ${\tilde{S}}^\theta_{k}(\tilde\lambda^\theta_k-2)$ is also unitary-valued, 
these two conditions impose a very strong constraint on the form this matrix.

If $\tilde \P_k^\theta$ is one dimensional, then one deduces that 
$\tilde{S}^\theta_{k}(\tilde\lambda^\theta_k-2)=\pm 1\tilde \P_k^\theta$. On the other hand, if $\tilde \P_k^\theta$ is two dimensional, then in the corresponding orthonormal basis, the 
matrix $\tilde{S}^\theta_{k}(\tilde\lambda^\theta_k-2)$ is given by a $2\times 2$ matrix
of the form \eqref{eq_form_mat}.
A similar argument holds for 
${\tilde{S}}^\theta_{k}(\tilde\lambda^\theta_k+2)$, starting with the restriction on 
$\uparrow^k$.
\end{proof}

Based on the these findings let us consider again the matrix-valued functions $s\mapsto \phi_{\downarrow_k}(s)$ and $s\mapsto \phi_{\uparrow^{k}}(s)$ exhibited in 
\eqref{eq:res1} and in \eqref{eq:res2}, and in particular let us compute their pointwise determinant. Some simple computations lead directly to:
\begin{enumerate}
\item[$(i)$] If $\tilde \P_k^\theta$ is one dimensional and $\tilde{\s}^\theta_{k}(\tilde\lambda^\theta_k-2)= 1$, then $\det\big(\phi_{\downarrow_k}(s)\big)=1$ while if  $\tilde{\s}^\theta_{k}(\tilde\lambda^\theta_k-2)= -1$, 
then $\det\big(\phi_{\downarrow_k}(s)\big)=\eta_-(s)$,
\item[$(ii)$] If $\tilde \P_k^\theta$ is one dimensional and $\tilde{\s}^\theta_{k}(\tilde\lambda^\theta_k+2)= 1$, then $\det\big(\phi_{\uparrow^{k}}(s)\big)=1$ while if  $\tilde{\s}^\theta_{k}(\tilde\lambda^\theta_k-2)= -1$, 
then $\det\big(\phi_{\uparrow^{k}}(s)\big)=\eta_+(s)$,
\item[$(iii)$] If $\tilde \P_k^\theta$ is two dimensional and $\tilde{\s}^\theta_{k}(\tilde\lambda^\theta_k-2) = \left(\begin{smallmatrix}1 & 0 \\  0 & 1 \end{smallmatrix}\right)$, then
$\det\big(\phi_{\downarrow_k}(s)\big)=1$, 
if $\tilde{\s}^\theta_{k}(\tilde\lambda^\theta_k-2)= 
-\left(\begin{smallmatrix}1 & 0 \\  0 & 1 \end{smallmatrix}\right)$, then
$\det\big(\phi_{\downarrow_k}(s)\big)=\eta_-(s)^2$, 
while if  $\tilde{\s}^\theta_{k}(\tilde\lambda^\theta_k-2) = 
\left(\begin{smallmatrix}a & b \\  \overline{b} & -a \end{smallmatrix}\right)$,
then $\det\big(\phi_{\downarrow_k}(s)\big)=\eta_-(s)$,
\item[$(iv)$] If $\tilde \P_k^\theta$ is two dimensional and $\tilde{\s}^\theta_{k}(\tilde\lambda^\theta_k+2) =
\left(\begin{smallmatrix}1 & 0 \\  0 & 1 \end{smallmatrix}\right)$, then
$\det\big(\phi_{\uparrow_k}(s)\big)=1$, 
if $\tilde{\s}^\theta_{k}(\tilde\lambda^\theta_k+2) =
-\left(\begin{smallmatrix}1 & 0 \\  0 & 1 \end{smallmatrix}\right)$, then
$\det\big(\phi_{\uparrow_k}(s)\big)=\eta_+(s)^2$, 
while if  $\tilde{\s}^\theta_{k}(\tilde\lambda^\theta_k+2) = 
\left(\begin{smallmatrix}a & b \\  \overline{b} & -a \end{smallmatrix}\right)$,
then $\det\big(\phi_{\uparrow_k}(s)\big)=\eta_+(s)$.
\end{enumerate}

The topological Levinson's theorem corresponds to an index theorem in scattering theory.
By considering the $C^*$-algebras introduced in Section \ref{sec:C*} we can consider the short
exact sequence of $C^*$-algebras
$$
0 \longrightarrow \K\big(\Hrond^\theta\big)\longrightarrow \EE^\theta \stackrel{q^\theta}{\longrightarrow}  \QQ^\theta \longrightarrow 0.
$$
Since $\F^\theta W_-^\theta(\F^\theta)^*$ is a Fredholm operator, when $\F^\theta W_-^\theta(\F^\theta)^*$ belongs to $\EE^\theta$
we infer the equality
\begin{equation}\label{eq:Lev1}
\ind\Big([q^\theta \big(\F^\theta W_-^\theta(\F^\theta)^*\big)]_1\Big)  
=-\big[\F^\theta E_{\rm p}(H^\theta)(\F^\theta)^*\big]_0,
\end{equation}
where $\ind$ denotes the index map from $K_1\big(\QQ^\theta\big)$ to $K_0\big( \K\big(\Hrond^\theta\big)\big)$ and where $E_{\rm p}(H^\theta)$ corresponds to the projection on the subspace spanned by the eigenfunctions of $H^\theta$.
Note that this projection appears from the standard relation 
$$
\big[1-(W^\theta_-)^* W_-^\theta\big]_0-\big[1-W_-^\theta (W_-^\theta)^*\big]_0 
=-\big[E_{\rm p}(H^\theta)\big]_0.
$$

The equality \eqref{eq:Lev1} can be directly deduced from \cite[Prop.~4.3]{Ri}.
Let us emphasize that this equality corresponds to the topological version of Levinson's theorem: it is a relation (by the index map) between the equivalence class in $K_1$
of quantities related to scattering theory, as described in Lemma \ref{lem:restrictions}, and the equivalence class in $K_0$ of the projection on the bound states of $H^\theta$. 
However, the standard formulation of Levinson's theorem is an equality between numbers. 
Thus, our final task is to extract a numerical equality from \eqref{eq:Lev1}.

In the next statement, the notation 
$\Var \big(\lambda \mapsto \det S^\theta(\lambda)\big)$ 
should be understood as  the total variation of the argument of the piecewise continuous function 
\begin{equation*}
I^\theta \ni \lambda \mapsto \det S^\theta(\lambda) \in \S^1.
\end{equation*}
where we compute the argument increasing with increasing $\lambda$. Our convention is also that the increase of the argument is counted clockwise.

\begin{Theorem}\label{thm:Adam}
For any $\theta \in [0,2\pi)$, but not in the intricate case,
the following equality holds:
\begin{equation}\label{eq:ouf}
\Var \big(\lambda \to \det S^\theta(\lambda)\big) + N - \tfrac{1}{2}C
 = \# \sigma_{\rm p}(H^\theta),
\end{equation}
where 
$$
C:=\#\big\{k\mid \tilde \s^\theta_{k}(\tilde \lambda^\theta_k\pm 2)=1 \big\} 
+ 2\#\big\{k\mid \tilde \s^\theta_{k}(\tilde \lambda^\theta_k\pm 2)= \left(\begin{smallmatrix}
1& 0 \\  0 & 1
\end{smallmatrix}\right)\big\}
+\#\Big\{k\mid \tilde \s^\theta_{k}(\tilde \lambda^\theta_k\pm 2)=\left(\begin{smallmatrix}
a& b \\  \overline{b} & -a \end{smallmatrix}\right) \Big\}.
$$
\end{Theorem}

\begin{proof}
The proof starts by evaluating both sides of \eqref{eq:Lev1} with 
the operator trace to obtain a numerical equation.
For the right hand side, we obtain (minus) the number of bound states, or more precisely
$\Tr\big(E_{\rm p}(H^\theta)\big)= \# \sigma_{\rm p}(H^\theta)$ if the multiplicity of the eigenvalues is taken into account.

For the left hand side, it is shown in \cite[Sec.~5]{APRR} that the algebra $\QQ^\theta$ can be naturally embedded (after rescaling) in $C_0\big(\R;\B(\C^N)\big)^+$. Thus, the index of $W_-^\theta$ is computed by the winding number of the pointwise determinant of $q^\theta(W_-^\theta)$, see for example \cite[Prop.~7]{KR08wind}. The various contributions for this computation can be inferred either from Figure
\ref{fig_Rt}, but the functions to be considered are provided by
Lemma \ref{lem:restrictions}. If we use the representation of the quotient algebra provided in 
Figure \ref{fig_Rt}, then all horizontal contributions $\phi_{\to j}$ can be encoded in the expression 
$\Var \big(\lambda \to \det S^\theta(\lambda)\big)$. 

For the vertical contributions, recall that the functions $\eta_\pm$ have been introduced in \eqref{eq_def_eta}.
These contributions have to be computed from $+\infty$
to $-\infty$ for the intervals with one endpoint at $\tl_k-2$, while they have to be computed from $-\infty$ to $+\infty$
for the intervals having an endpoint at $\tl_k+2$. 

In both cases, if $\tilde \s^\theta_{k}(\tilde\lambda^\theta_k\pm 2)=1$, 
then the contribution is $0$, since 
$\det\big(\phi_{\downarrow_k}(s)\big)=1$ and $\det\big(\phi_{\uparrow_k}(s)\big)=1$
according to the computation performed before the statement.
Similarly,  if $\tilde \s^\theta_{k}(\tilde\lambda^\theta_k\pm 2)=\left(\begin{smallmatrix}
1& 0 \\  0 & 1 \end{smallmatrix}\right)$, 
then the contribution is $0$.
On the other hand, if  $\tilde \s^\theta_{k}(\tilde\lambda^\theta_k\pm2)=-1$,
then the contribution is of $\frac{1}{2}$, with our clockwise convention for the increase of the variation.
Similarly, if  $\tilde \s^\theta_{k}(\tilde\lambda^\theta_k\pm 2)=-\left(\begin{smallmatrix}
1& 0 \\  0 & 1 \end{smallmatrix}\right)$, then 
the contribution is equal to $1$.
Finally, if $\tilde \s^\theta_{k}(\tilde\lambda^\theta_k\pm 2)= \left(\begin{smallmatrix}
a& b \\  \overline{b} & -a \end{smallmatrix}\right)$, the contribution is $\frac{1}{2}$.
Note that our presentation of \eqref{eq:ouf} takes into account that $\tilde \s^\theta_k(\tilde \lambda^\theta_k\pm 2)=-1$ holds
generically, see \cite[Lem.~4.3]{APRR}. 
\end{proof}

\section{Intricate case}\label{sec_Degenerate}
\setcounter{equation}{0}
Let us recall that the \emph{intricate case} stands for the case where $\theta=0$,
$N\in2\N$, and $\v\xi_N^0$ and $\v\xi^0_{N/2}$ are linearly dependent, 
see Definition \ref{def_intricate}.
This special case was left untouched in \cite{NRT} and in \cite{APRR}.
Indeed, it was then impossible to show that the remainder term in \eqref{eq_formula}
is compact, and the affiliation of the wave operators to the algebra $\EE^0$
could not be proved.
In this section, we provide a full analysis of this situation, and therefore we tacitly assume
that $N$ is even and that $\theta=0$.

\subsection{New integral operators}\label{sec_extra_t}

In the intricate case, the problem appears only in two terms of the expression provided in \eqref{eq_remainder}. By taking the limit $\varepsilon\searrow 0$ and by using \cite[Lem.~4.1]{NRT}, the two annoying terms read for $\zeta,\zeta'\in\Drond^0$
\begin{align}
\label{eq_anno1} & -\frac{1}{\pi}\int_0^4\bigg\langle
\P_{N/2}^0\v M^0(\lambda+i0)\v\;\!\P_{N}^0
\zeta'_{N}(\lambda),\int_{-4}^0\frac{
\beta_{N}^0(\lambda)^{-1}\beta_{N/2}^0(\mu)^{-1}}{\mu-\lambda}\;\!\zeta_{N/2}(\mu)
\;\!\d\mu\bigg\rangle_{\C^N}\d\lambda \\
\label{eq_anno2} & - \frac{1}{\pi}
\int_{-4}^0\bigg\langle
\P_N^0\v M^0(\lambda+i0)\v\;\!\P_{N/2}^0
\zeta'_{N/2}(\lambda),\int_{0}^4\frac{
\beta_{N/2}^0(\lambda)^{-1}\beta_N^0(\mu)^{-1}}{\mu-\lambda}\;\!\zeta_N(\mu)
\;\!\d\mu\bigg\rangle_{\C^N}\d\lambda,
\end{align}
while all other terms are compact, in the intricate case and in the non-intricate case,
see \cite[Prop.~4.8 \& 4.12]{NRT}.

Let us now concentrate on the first of these two terms. For that purpose, we define a new
integral operator whose kernel is given for $\lambda\in (0,4)$ and $\mu \in (-4,0)$ by
$$
\Theta(\lambda,\mu):=-\frac{1}{\pi}\  \frac{\beta^0_{N/2}(\lambda)^2 \ \beta^0_N(\lambda)^{-1} \ \beta^0_{N/2}(\mu)^{-1}}{\mu-\lambda}.
$$
With this notation, the expression \eqref{eq_anno1} reads
$$
\int_0^4\bigg\langle
\beta^0_{N/2}(\lambda)^{-2} \P_{N/2}^0\v M^0(\lambda+i0)\v\;\!\P_{N}^0
\zeta'_{N}(\lambda),\int_{-4}^0\Theta(\lambda,\mu)\;\!\zeta_{N/2}(\mu)
\;\!\d\mu\bigg\rangle_{\C^N}\d\lambda.
$$
It has been shown in \cite[Lem.~4.5]{NRT} that the map
$$
(0,4)\ni \lambda \mapsto  
\beta^0_{N/2}(\lambda)^{-2} \P_{N/2}^0\v M^0(\lambda+i0)\v\;\!\P_{N}^0 \in \B(\C^N)
$$
extends to a continuous (and therefore bounded) function on $[0,4]$.
This bounded matrix-valued multiplication operator is denoted by $Q^0_{N/2,N}(X^0)$.

For the study of the operator $\Theta$, we shall look at a different representation.
Before this, observe that
\begin{align*}
\Theta(\lambda,\mu)
& =  -\frac{1}{\pi}\frac{|(\lambda+2)^2-4|^{1/2}\ |(\lambda-2)^2-4|^{-1/4}\ 
|(\mu+2)^2-4)^{-1/4}}{\mu-\lambda} \\
& = -\frac{1}{\pi} |\lambda+4|^{1/2}\frac{1}{|\lambda-4|^{1/4}}
\frac{|\lambda|^{1/4}\ |\mu|^{-1/4}}{\mu-\lambda} \frac{1}{|\mu+4|^{1/4}}.
\end{align*}
For the next statement, recall that $\V_j^0$ has been defined in \eqref{eq:Vjt}.

\begin{Lemma}
In $\ltwo(\R)$, the following equality holds for $s,t\in \R$
$$
\big[\V^0_N \Theta(\V^0_{N/2})^*\big](s,t)=
\frac{1}{\sqrt{2}\pi}\big(3+\tanh(s)\big)^{1/2}\;\!
\e^{(s+t)/2}\sech(s+t)\;\! \chi_+(t)\;\! b(t) + k(s,t),
$$
where $\chi_+\in C^\infty(\R)$ satisfies $\chi_+(t)=0$ for $t<-1$ and $\chi_+(t)=1$ for $t>1$, 
where
$$
b(t):=\big(\e^t+\e^{-t}\big)^{1/2}\big(\e^{t/2}+\e^{-t/2}\big)^{-1},
$$
and where $k(s,t)$ denotes the integral kernel of a Hilbert-Schmidt operator.
\end{Lemma}

\begin{proof}
By direct computations, one gets that
\begin{align*}
& \big[\V^0_N \Theta(\V^0_{N/2})^*\big](s,t)\\
& = -\frac{1}{\pi}\big(3+\tanh(s)\big)^{1/2} \left(\frac{1+\tanh(s)}{\cosh(s)}\right)^{1/2}
\frac{1}{-2+\tanh(t)-\tanh(s)}\left(\frac{1}{\cosh(t)}\right)^{1/2} \\
& = -\frac{1}{\pi}\big(3+\tanh(s)\big)^{1/2} \big(1+\tanh(s)\big)^{1/2}
\frac{\cosh(s)^{1/2}\cosh(t)^{1/2}}{-2\cosh(t)\cosh(s)+\sinh(t-s)} \\
& =  -\frac{1}{\pi}\big(3+\tanh(s)\big)^{1/2} \e^{s/2}
\frac{\e^{t/2}+\e^{-t/2}}{-2\cosh(t)\cosh(s)+\sinh(t-s)}\;\!\frac{1}{\sqrt{2}} \;\!b(t) \\
& =  \frac{1}{\sqrt{2}\pi}\big(3+\tanh(s)\big)^{1/2} 
\frac{\e^{(s+t)/2}+\e^{(s-t)/2}}{\cosh(s+t)+\e^{s-t}} \;\!b(t). 
\end{align*}
Since the multiplication operators defined by the functions
$s\mapsto \big(3+\tanh(s)\big)^{1/2}$ and $t\mapsto b(t)$
are bounded, the statement holds if one shows that the following kernels correspond
to Hilbert-Schmidt operators:
$$
R_1(s,t):=\frac{\e^{(s-t)/2}}{\cosh(s+t)+\e^{s-t}}
$$
and
$$
R_2(s,t):=\frac{\e^{(s+t)/2}}{\cosh(s+t)+\e^{s-t}} - \e^{(s+t)/2}\sech(s+t)\;\! \chi_+(t).
$$

In order to compute the $\ltwo$-norm of $R_1$ and of $R_2$, one looks at these expressions in polar coordinates. Namely, for $R_1$ and for $r>0$ and $\omega \in [0,2\pi]$ one considers  the expression
$$
R_1\big(r\sin(\omega),r\cos(\omega)\big)=
\frac{1}{\e^{r(\cos(\omega)-\sin(\omega))/2}\cosh(r(\cos(\omega)+\sin(\omega)))+\e^{-r(\cos(\omega)-\sin(\omega))/2}}.
$$
By considering the equalities $\cos(\alpha)\pm\sin(\alpha)=\sqrt{2}\cos(\alpha\mp\pi/4)$
one infers that
the denominator of the previous expression satisfies
\begin{align*}
& \e^{r(\cos(\omega)-\sin(\omega))/2}\cosh(r(\cos(\omega)+\sin(\omega)))+\e^{-r(\cos(\omega)-\sin(\omega))/2} \\
& =\e^{\sqrt{2}r\cos(\omega+\pi/4)/2}\cosh\big(\sqrt{2}r\cos(\omega-\pi/4)\big)+\e^{-\sqrt{2}r\cos(\omega+\pi/4)/2}.
\end{align*}
Then, one observe that one of the two terms $\e^{\pm\sqrt{2}r\cos(\omega+\pi/4)/2}$ is exponentially growing (as a function of $r$) with a constant independent of $\omega$, except in a neighborhood of $\omega=\pi/4$ and of $\omega=5\pi/4$.
However, in these neighborhoods, the factor $\cosh\big(\sqrt{2}r\cos(\omega-\pi/4)\big)$ is exponentially growing. Thus, there exist $c>0$ and $d>0$ independent of $\omega$ such that
$$
R_1\big(r\sin(\omega),r\cos(\omega)\big)\leq c\e^{-dr}
$$
for all $\omega \in [0,2\pi]$. 

For $R_2$, observe that
\begin{align*}
R_2(s,t) & = \frac{(1-\chi_+(t))\cosh(s+t)-\chi_+(t)\e^{s-t}}{\e^{-(s+t)/2}\cosh(s+t)[\cosh(s+t)+\e^{s-t}]} \\
& = \frac{1-\chi_+(t)}{\e^{-(s+t)/2}[\cosh(s+t)+\e^{s-t}]}-\frac{\chi_+(t)\e^{s-t}}{\e^{-(s+t)/2}\cosh(s+t)[\cosh(s+t)+\e^{s-t}]}.
\end{align*}
For the first term and by using the same approach, observe that its denominator
can be rewritten as
\begin{align*}
&\e^{-(s+t)/2}[\cosh(s+t)+\e^{s-t}] \\
&=\e^{-\sqrt{2}r\cos(\omega-\pi/4)/2}
\big[\cosh\big(\sqrt{2}r \cos(\omega-\pi/4)\big)+\e^{-\sqrt{2}r\cos(\omega+\pi/4)}\big].
\end{align*}
In addition, because of the cut-off $1-\chi_+$ we can restrict our attention
to $\omega \in [\pi/2-\varepsilon,3\pi/2+\varepsilon]$ for some fixed $\varepsilon>0$
small enough.  In this region, the product 
$$
\e^{-\sqrt{2}r\cos(\omega-\pi/4)/2} \cosh\big(\sqrt{2}r \cos(\omega-\pi/4)\big)
$$ 
is exponentially growing (as a function of $r$) with a constant independent of $\omega$, except in a neighborhood of $\omega=3\pi/4$.
However, in any such neighborhood, the product
$$
\e^{-\sqrt{2}r\cos(\omega-\pi/4)/2}\e^{-\sqrt{2}r\cos(\omega+\pi/4)}
$$
is exponentially growing (as a function of $r$) with a constant independent of $\omega$.
One then infers that the first term is $R_2$ is smaller than $c\e^{-dr}$ for some
$c,d>0$ and all $\omega \in [\pi/2-\varepsilon,3\pi/2+\varepsilon]$.

The second term in $R_2$ is treated similarly, for $\omega$ restricted to $[\pi/2+\varepsilon,3\pi/2-\varepsilon]$ with  $\varepsilon>0$
small enough.  We consider
\begin{equation}\label{eq_HS1}
\frac{\e^{-\sqrt{2}r\cos(\omega+\pi/4)}}{\e^{-\sqrt{2}r\cos(\omega-\pi/4)/2}\cosh(\sqrt{2}r\cos(\omega-\pi/4))[\cosh(\sqrt{2}r\cos(\omega-\pi/4))+\e^{-\sqrt{2}r\cos(\omega+\pi/4)}]}.
\end{equation}
Since
$$
0<\frac{\e^{-\sqrt{2}r\cos(\omega+\pi/4)}}{\cosh(\sqrt{2}r\cos(\omega-\pi/4))+\e^{-\sqrt{2}r\cos(\omega+\pi/4)}}\leq 1
$$
the exponential decay of \eqref{eq_HS1} (as a function of $r$) with a constant independent of $\omega$ is obtained by the first two factors of the denominator, except in a neighborhood of $\omega=7\pi/4$. In any such neighborhood, observes that
$$
\cosh(\sqrt{2}r\cos(\omega-\pi/4))\big[\cosh\big(\sqrt{2}r\cos(\omega-\pi/4)\big)+\e^{-\sqrt{2}r\cos(\omega+\pi/4)}\big]\geq 1
$$
while 
$$
\frac{\e^{-\sqrt{2}r\cos(\omega+\pi/4)}}{\e^{-\sqrt{2}r\cos(\omega-\pi/4)/2}}
= \e^{-r\big(\sqrt{2}\cos(\omega+\pi/4)-\sqrt{2}\cos(\omega-\pi/4)/2\big)}\leq \e^{-dr}
$$
for some $d>0$ and independent of $\omega$.

By collecting the various estimates, one obtains that the functions $R_1$ and $R_2$ belong 
to $\ltwo(\R^2)$, and therefore the corresponding kernel $k$ in the statement is a Hilbert-Schmidt operator.
\end{proof}

Let us further study the operator in $\ltwo(\R)$ whose kernel is given for $s,t\in \R$ by
$$
\frac{1}{\sqrt{2}\pi}\big(3+\tanh(s)\big)^{1/2}\;\!
\e^{(s+t)/2}\sech(s+t)\;\! \chi_+(t)\;\! b(t).
$$
By setting the unitary and self-adjoint operator $I\in \B\big(\ltwo(\R)\big)$ with $[If](x):=f(-x)$, 
the corresponding operator can be rewritten as
\begin{align*}
& \frac{1}{\sqrt{\pi}} \big(3+\tanh(X)\big)^{1/2} \psi(D)\;\!I\;\!\chi_+(X) \;\!b(X) \\
= & \frac{1}{\sqrt{\pi}} \big(3+\tanh(X)\big)^{1/2} \psi(D)\;\!\chi_+(-X)\;\!b(-X)\;\! I
\end{align*}
with 
$$
[\psi(D)f](x)=[(\F^{-1}\psi)*f](x) = \frac{1}{\sqrt{2\pi}}\int_\R [\F^{-1}\psi](x-y)f(y) \d y
$$
and 
\begin{equation}\label{eq_psi}
[\F^{-1}{\psi}](x):=\e^{x/2}\;\!\sech(x).
\end{equation}
Here $\F$ denotes the Fourier transform on $\R$.
Observe that $\psi\in C_0(\R)$ since the function $\F^{-1}\psi$ belongs to $\lone(\R)$.
Thus, by using the result of Cordes on the compactness of commutators of functions of $X$ and $D$ when the functions have limits at $\pm \infty$, see for instance \cite[Thm.~4.1.10]{ABG}, one infers that
\begin{equation}\label{eq_op1}
\V^0_N \Theta(\V^0_{N/2})^*=
\frac{1}{\sqrt{\pi}} \chi_+(-X) \;\!b(-X)\;\! \big(3+\tanh(X)\big)^{1/2} \;\!\psi(D) \;\!I + K
\end{equation}
with $K$ a compact operator in $\ltwo(\R)$.

Let us now provide the explicit expression for the function $\psi$, whose Fourier transform
appears in \eqref{eq_psi}.

\begin{Lemma}
The Fourier transform of the $\lone(\R)$-function
$$
\R\ni x\mapsto \e^{x/2} \sech(x)\in \R
$$
is the continuous function $\psi:\R\to \C$ given for $y\in \R$ by
$$
\psi(y):=\sqrt{\pi}\ \frac{\cosh(\pi y/2)-i\sinh(\pi y/2)}{\cosh(\pi y)}.
$$
\end{Lemma}

\begin{proof}
By a change of variables and by using the Laplace transform of the hyperbolic secant function, 
as given by \cite[Eq.~4.9.7]{BatemanI}, one infers that
\begin{align*}
& \int_\R  \e^{x/2} \sech(x)\e^{-ixy}\d x  \\
& = \int_\R \e^{-x(-\frac{1}{2}+iy)}\sech(x) \d x \\
& =  \int_0^\infty \e^{-x(-\frac{1}{2}+iy)}\sech(x) \d x 
+ \int_0^\infty \e^{-x(\frac{1}{2}-iy)}\sech(x) \d x \\
& = \frac{1}{2}
\Big(
\psi\big(\tfrac{5}{8}+i\tfrac{y}{4}\big)-\psi\big(\tfrac{1}{8}+i\tfrac{y}{4}\big)
+ \psi\big(\tfrac{7}{8}-i\tfrac{y}{4}\big)-\psi\big(\tfrac{3}{8}-i\tfrac{y}{4}\big)
\Big),
\end{align*}
where $\psi$ denotes the digamma function. By the reflection formula
$\psi(1-z)-\psi(z)=\pi \cot(\pi z)$ for $z\in \C$ and $\cot$ the cotangent function, this expression simplifies into
\begin{equation}\label{eq_cot}
-\frac{\pi}{2}\Big(\cot\big(\tfrac{5\pi}{8}+i\tfrac{\pi y}{4}\big)
+ \cot\big(\tfrac{7\pi}{8}-i\tfrac{\pi y}{4}\big)\Big).
\end{equation}
Considering finally the equality \cite[Eq.~4.3.58]{AS} for $a,b\in \R$
$$
\cot(a+ib) = \frac{\sin(2a)-i\sinh(2b)}{\cosh(2b)-\cos(2a)}
$$
one deduces that \eqref{eq_cot} is equal to
\begin{align*}
&-\frac{\pi}{2}\bigg(
\frac{\sin(5\pi/4)-i\sinh(\pi y/2)}{\cosh(\pi y/2)-\cos(5\pi/4)}
+\frac{\sin(7\pi/4)+i\sinh(\pi y/2)}{\cosh(\pi y/2)-\cos(7\pi/4)}
\bigg) \\
& = -\frac{\pi}{2}\ \frac{-\sqrt{2}\cosh(\pi y/2)+i\sqrt{2}\sinh(\pi y/2)}{\cosh(\pi y/2)^2-1/2} \\
& = \sqrt{2}\pi \ \frac{\cosh(\pi y/2)-i\sinh(\pi y/2)}{\cosh(\pi y)} 
\end{align*}
where the equality $\cosh(2z)=2\cosh(z)^2-1$ has been used for the last equality.
By implementing the correct normalization of the initial Fourier transform, one gets the result.
\end{proof}

Before studying \eqref{eq_op1} in more detail, let us observe that the second expression
\eqref{eq_anno2} leads to a similar operator.
Indeed, let us set  $Q^0_{N,N/2}(X^0)$ for the bounded operator
defined by the continuous extension to $[-4,0]$ of the function
$$
(-4,0)\ni \lambda \mapsto  
\beta^0_{N}(\lambda)^{-2} \P_{N}^0\v M^0(\lambda+i0)\v\;\!\P_{N/2}^0 \in \B(\C^N),
$$
and let us observe that for $\lambda\in (-4,0)$ and $\mu\in (0,4)$ one has
\begin{align*}
\tilde \Theta(\lambda,\mu)
& =- \frac{1}{\pi}\  \frac{\beta^0_{N}(\lambda)^2 \ \beta^0_{N/2}(\lambda)^{-1} \ \beta^0_{N}(\mu)^{-1}}{\mu-\lambda} \\
& =  -\frac{1}{\pi}\frac{|(\lambda-2)^2-4|^{1/2}\ |(\lambda+2)^2-4|^{-1/4}\ 
|(\mu-2)^2-4)^{-1/4}}{\mu-\lambda} \\
& = -\frac{1}{\pi} |\lambda-4|^{1/2}\frac{1}{|\lambda+4|^{1/4}}
\frac{|\lambda|^{1/4}\ |\mu|^{-1/4}}{\mu-\lambda} \frac{1}{|\mu-4|^{1/4}} \\
& = -\Theta(-\lambda, -\mu).
\end{align*}
If we define the unitary operator 
\begin{equation}\label{eq_iota}
\iota: \ltwo(0,4)\to \ltwo(-4,0),  \quad 
[\iota \zeta](\lambda):=\zeta(-\lambda) \quad \forall \zeta \in \ltwo(0,4),
\lambda \in (-4,0),
\end{equation}
 then one
readily observes that
$$
\V_{N/2}^0\tilde \Theta (\V_N^0)^* 
= - \V_{N/2}^0\iota \Theta \iota(\V_N^0)^*  
= - I \;\!\V_{N}^0 \Theta (\V_{N/2}^0)^*\;\!I, 
$$
which implies that
\begin{equation}\label{eq_op2}
\V_{N/2}^0\tilde \Theta (\V_N^0)^* =
- \frac{1}{\sqrt{\pi}} I\;\!\chi_+(-X) \;\!b(-X)\;\! \big(3+\tanh(X)\big)^{1/2} \;\!\psi(D) + K
\end{equation}
with $K$ a compact operator in $\ltwo(\R)$.

\subsection{New operators and their limits}

Let us now recast the computation of the previous section in their context.
We firstly recall that $N$ is an even number. 
The operators defined by \eqref{eq_anno1} and \eqref{eq_anno2}
can be rewritten as $\Theta^* Q^0_{N/2,N}(X^0)$ and as
$\tilde \Theta^* Q^0_{N,N/2}(X^0)$, the former operator acting from
$\ltwo\big((0,4);\P^0_N \C^N\big)$ to
$\ltwo\big((-4,0);\P^0_{N/2} \C^N\big)$, while the latter acts from
$\ltwo\big((-4,0);\P^0_{N/2} \C^N\big)$ to 
$\ltwo\big((0,4);\P^0_N \C^N\big)$. These spaces are in fact part of the Hilbert
space $\Hrond^0$ introduced in \eqref{eq_Hrond}.

In order to have the correct framework, let us consider
the Hilbert space
\begin{equation}\label{eq_def_int}
\Hrond^0_{\rm int}:=\ltwo(-4,0) \oplus  \ltwo(0,4)
\end{equation}
which is naturally identified with 
\begin{equation}\label{eq_Hil}
\ltwo\big((-4,0);\P^0_{N/2} \C^N\big) \, \oplus  \, \ltwo\big((0,4);\P^0_N \C^N\big)
\subset \Hrond^0.
\end{equation}
The next statement collects the expressions obtained so far
for the restriction of the wave operator $W_-^0$ to this subspace. 
Note that we shall use the notation
$$
 \Pi_{j,j}:=\tfrac12\big\{(\V^{0}_j)^*\big(1-\tanh(\pi D)-i\cosh(\pi D)^{-1}\tanh(X)\big)
\V^0_j \big\} \in \B\big(\ltwo(I^0_j)\big).
$$
We also define the scalar-valued function $\s_{j,j}^0$ by the relation 
$S^0(\lambda)_{j,j}=:\s_{j,j}^0(\lambda)\otimes \P_j^0$. 
Finally, by using the bra-ket notation, we can also set
$$
Q^0_{N/2,N}(X^0) =\q^0_{N/2,N}(X^0) \otimes \tfrac{1}{N}|\xi^0_{N/2}\rangle\langle \xi^0_N | 
$$
and 
$$
Q^0_{N,N/2}(X^0) = \q^0_{N,N/2}(X^0)  \otimes \tfrac{1}{N}|\xi^0_{N}\rangle\langle \xi^0_{N/2} |
$$
with $\q^0_{N/2,N}(X^0)$ and $\q^0_{N,N/2}(X^0)$ two scalar-valued functions.

\begin{Proposition}\label{prop_restriction}
Let $N$ be even, consider the restriction of $\F^0 W^0_-(\F^0)^*$
to the subspace \eqref{eq_Hil}, and let $(\W_-^0)_{\rm int}$ 
be this operator once the Hilbert space is identified with $\Hrond^0_{\rm int}$.
In the non-intricate case, this operator
takes the form
$$
\left(\begin{matrix}
1+ \Pi_{N/2,N/2}\;\!\big(\s^0_{N/2,N/2}(X^0)-1\big) & 0 \\
0 & 1+ \Pi_{N,N}\;\!\big(\s^0_{N,N}(X^0)-1\big)
\end{matrix}\right)
+ K
$$
with $K \in \K\big(\Hrond^0_{\rm int}\big)$, 
while in the intricate case this operator takes the form
$$
\left(\begin{matrix}1+ \Pi_{N/2,N/2}\;\!\big(\s^0_{N/2,N/2}(X^0)-1\big)
& \Theta^* \q^0_{N/2,N}(X^0) \\
\tilde \Theta^* \q^0_{N,N/2}(X^0) & 1+ \Pi_{N,N}\;\!\big(\s^0_{N,N}(X^0)-1\big)
\end{matrix}\right)
+K
$$
with $K\in \K\big(\Hrond^0_{\rm int}\big)$.
\end{Proposition}

The proof of this statement is straightforward once it is observed
that the operators $S^0_{N,N/2}$ and $S^0_{N/2,N}$ do not exist, since the
two intervals $I^0_{N}$ and $I^0_{N/2}$ have an empty intersection.

Let us now define the unitary map 
$\V^0_{\rm int}: \Hrond^0_{\rm int} \to \ltwo(\R;\C^2)$ given by
$\V^0_{\rm int}:=\V_{N/2}^0\oplus \V_N^0$.
The operators exhibited in the previous statement can then be represented in 
$\ltwo(\R;\C^2)$ by conjugating them with this unitary transformation.
In particular in the intricate case one gets
\begin{align}\label{eq_def_block}
\nonumber & \V^0_{\rm int}\left(\begin{matrix}1+ \Pi_{N/2,N/2}\;\!\big(\s^0_{N/2,N/2}(X^0)-1\big)
& \Theta^* \q^0_{N/2,N}(X^0) \\
\tilde \Theta^* \q^0_{N,N/2}(X^0) & 1+ \Pi_{N,N}\;\!\big(\s^0_{N,N}(X^0)-1\big)
\end{matrix}\right) (\V^0_{\rm int})^* \\
& =:
\left(\begin{matrix} A_{N/2,N/2} & A_{N/2,N} \\
A_{N,N/2} & A_{N,N}
\end{matrix}\right)
\end{align}
with 
\begin{align*}
& A_{N/2,N/2}:=1+ \frac{1}{2}\big(1-\tanh(\pi D)-i\cosh(\pi D)^{-1}\tanh(X)\big)\;\!\big(\s^0_{N/2,N/2}\big(-2+2\tanh(X)\big)-1\big), \\
& A_{N/2,N}:= \big(\V^0_N \Theta(\V^0_{N/2})^*\big)^* \q^0_{N/2,N}\big(2+2\tanh(X)\big) \\
& A_{N,N/2}:= \big(\V_{N/2}^0\tilde \Theta (\V_N^0)^* \big)^* \q^0_{N,N/2}\big(-2+2\tanh(X)\big) \\
& A_{N,N}:=1+ \frac{1}{2}\big(1-\tanh(\pi D)-i\cosh(\pi D)^{-1}\tanh(X)\big)\;\!\big(\s^0_{N,N}\big(2+2\tanh(X)\big)-1\big).
\end{align*}
Clearly, a similar and simpler formula holds in the non-intricate case.

Let us now further concentrate on the operators $A_{N/2,N}$ and $A_{N,N/2}$, 
starting with $A_{N/2,N}$. It follows from \eqref{eq_op1} that 
\begin{equation}\label{eq_MA12}
A_{N/2,N} = I\overline{\psi}(D) m_{N/2,N}(X) + K
\end{equation}
with $K\in \K\big(\ltwo(\R)\big)$ and with
$$
m_{N/2,N}(X) := \frac{1}{\sqrt{\pi}} \chi_+(-X) \;\!b(-X) \big(3+\tanh(X)\big)^{1/2} \q^0_{N/2,N}\big(2+2\tanh(X)\big). 
$$
For later use, let us directly observe that
\begin{align*}
m_{N/2,N}(\infty) & = 0, \\
m_{N/2,N}(-\infty) & = \frac{\sqrt{2}}{\sqrt{\pi}}\q^0_{N/2,N}(0_+), \\
\big(m_{N/2,N}\big)_{\rm e}(\infty) & = \frac{1}{\sqrt{2\pi}}\q^0_{N/2,N}(0_+), \\
\big(m_{N/2,N}\big)_{\rm o}(\infty) & = -\frac{1}{\sqrt{2\pi}}\q^0_{N/2,N}(0_+), 
\end{align*}
where the indices $\rm e$ and $\rm o$ stand for the even and the odd part of a function defined on $\R$.
It therefore remains to compute the expression for $\q^0_{N/2,N}(0_+)$.
Note that the next statement holds in the intricate case only, 
and that its proof involves some operators defined in \cite{NRT} only.

\begin{Lemma}\label{lem_limit1}
Let $\alpha\in \{-1,1\}$ satisfying $\v\xi^0_{N/2}=\alpha \v \xi^0_N$. Then one has
$$
\q^0_{N/2,N}(0_+) = -\frac{1}{2}(1+i)\alpha.
$$
\end{Lemma}

\begin{proof}
It has been shown before in \cite[Lemma 4.10]{NRT} that $Q^0_{N/2,N}(0_+)= -\frac{i}{2}\P^0_{N/2}\v I_0(0)^{-1}\v\P^0_N$
with $2 I_0(0)=\v \P^0_N\v+i\v \P^0_{N/2}\v$.
By the colinearity of the two vectors one infers that
\begin{align*}
\v \P^0_N\v+i\v \P^0_{N/2}\v 
& = \frac{1}{N} \Big(\big|\v\xi_{N}^0\big\rangle \big\langle \v\xi_{N}^0\big|
+ i \big|\v\xi_{N/2}^0\big\rangle \big\langle \v\xi_{N/2}^0\big|\Big) \\
& = \frac{1}{N} (1+i)  \big|\v\xi_{N}^0\big\rangle \big\langle \v\xi_{N}^0\big| \\
& = \frac{1}{N}\|v\|_{\ltwo}^2 (1+i)  \big|\widehat{\v\xi_{N}^0}\big\rangle \big\langle \widehat{\v\xi_{N}^0}\big|
\end{align*}
leading to $I_0(0)=\frac{1}{2N} \|v\|_{\ltwo}^2 (1+i)  \big|\widehat{\v\xi_{N}^0}\big\rangle \big\langle \widehat{\v\xi_{N}^0}\big|$.
As a consequence, one has
\begin{align*}
-\frac{i}{2}\P^0_{N/2}\v I_0(0)^{-1}\v\P^0_N
& = -\frac{i}{2}\frac{1}{N^2}\big|\xi_{N/2}^0\big\rangle \big\langle \v\xi_{N/2}^0\big|
\frac{2N}{\|v\|_{\ltwo}^2 (1+i)} \Big( \big|\widehat{\v\xi_{N}^0}\big\rangle \big\langle \widehat{\v\xi_{N}^0}\big|\Big)^{-1}
\big|\v\xi_{N}^0\big\rangle \big\langle \xi_{N}^0\big| \\
& = -\frac{i}{2}(1-i) \alpha \frac{1}{N}\big|\xi_{N/2}^0\big\rangle \big\langle \widehat{\v\xi_{N}^0}\big|
\Big( \big|\widehat{\v\xi_{N}^0}\big\rangle \big\langle \widehat{\v\xi_{N}^0}\big|\Big)^{-1}
\big|\widehat{\v\xi_{N}^0}\big\rangle \big\langle \xi_{N}^0\big| \\
& = -\frac{1}{2}(1+i) \alpha \big|\widehat{\xi_{N/2}^0}\big\rangle  \big\langle \widehat{\xi_{N}^0}\big|,
\end{align*}
as expected.
\end{proof}

As a consequence of the previous statement one infers that in the intricate case one has
\begin{equation*}
\big(m_{N/2,N}\big)_{\rm e}(\infty) = -\frac{1}{2\sqrt{2\pi}}(1+i)\alpha, 
\qquad \big(m_{N/2,N}\big)_{\rm o}(\infty) = \frac{1}{2\sqrt{2\pi}}(1+i)\alpha. 
\end{equation*}

In a similar way, the expression for $A_{N,N/2}$ can be computed.
From \eqref{eq_op2} one infers that
\begin{equation}\label{eq_MA21}
A_{N,N/2}=  -\overline{\psi}(D) m_{N,N/2}(X) I + K,
\end{equation}
with $K\in \K\big(\ltwo(\R)\big)$ and with
$$
m_{N,N/2}(X):= \frac{1}{\sqrt{\pi}}\;\!\chi_+(-X) \;\!b(-X)\;\! \big(3+\tanh(X)\big)^{1/2} \q^0_{N,N/2}\big(-2+2\tanh(-X)\big).
$$
As before one readily observes that
\begin{align*}
m_{N,N/2}(\infty) & = 0, \\
m_{N,N/2}(-\infty) & = \frac{\sqrt{2}}{\sqrt{\pi}}\q^0_{N/2,N}(0_-), \\
\big(m_{N,N/2}\big)_{\rm e}(\infty) & = \frac{1}{\sqrt{2\pi}}\q^0_{N/2,N}(0_-), \\
\big(m_{N,N/2}\big)_{\rm o}(\infty) & = -\frac{1}{\sqrt{2\pi}}\q^0_{N/2,N}(0_-). 
\end{align*}
In addition, as in Lemma \ref{lem_limit1} one can check that 
for $\alpha\in \{-1,1\}$ satisfying $\v\xi^0_{N/2}=\alpha \v \xi^0_N$  one has
$\q^0_{N,N/2}(0_-) = \frac{1}{2}(1-i)\alpha$.
Consequently, it follows that 
$$
\big(m_{N,N/2}\big)_{\rm e}(\infty) = \frac{1}{2\sqrt{2\pi}}(1-i)\alpha, \qquad 
\big(m_{N,N/2}\big)_{\rm o}(\infty) = -\frac{1}{2\sqrt{2\pi}}(1-i)\alpha. 
$$

\subsection{Even / odd representation and Cordes algebra}\label{sec_Cordes_alg}

In this section we develop some preliminaries for the $C^*$-algebraic framework
fully developed in the subsequent sections. As already emphasized, the main difference with the 
cases already studied in \cite{APRR} is the non-compactness 
of the terms $A_{N/2,N}$ and $A_{N,N/2}$ exhibited in the previous section.
Not only these two terms are not compact, but in addition they are not only defined
by functions of $X$ and $D$, see \eqref{eq_MA12} and \eqref{eq_MA21}: they also contain
the operator $I$ which can not be simply expressed in terms of $X$ and $D$.
In order to accommodate this new operator, the framework has to be slightly enlarged:
We need to consider one more unitary 
transformation, namely the decomposition into even and odd functions on $\ltwo(\R)$.
It is then easily observed that the operator $I$ acts in a very simple form on even and on odd functions.

Let us consider 
$\U:\ltwo(\R)\to\ltwo(\R_+;\C^2)$ given by
$$
\U f:=\sqrt2
\begin{pmatrix}
f_{\rm e}\\
f_{\rm o}
\end{pmatrix}\quad{\rm and}\quad
\big[\U^*
\big(\begin{smallmatrix}
f_1\\
f_2
\end{smallmatrix}\big)\big](x)
:=\textstyle\frac1{\sqrt2}
\big[f_1(|x|)+\sgn(x)f_2(|x|)\big],
$$
for $f\in\ltwo(\R)$, $\big(\begin{smallmatrix}
f_1\\
f_2
\end{smallmatrix}\big)\in\ltwo(\R_+;\C^2)$, and $x\in\R$.
Here $f_{\rm e}, f_{\rm o}$ denote the even and the odd part of $f$.
Then, one observes that if $m$ is a function on $\R$ 
\begin{equation*}
\U m(X)\U^*=
\left(\begin{smallmatrix}
m_{\rm e}(L)~ & m_{\rm o}(L)\\
m_{\rm o}(L)~ & m_{\rm e}(L)
\end{smallmatrix}\right),
\end{equation*}
where $L$ stands for the multiplication operator by the variable in $\ltwo(\R_+)$.

In order to consider $\U m(D)\U^*$, let us denote by $\F_1$ the usual unitary Fourier transform
in $\ltwo(\R)$, and let $\F_{\rm N}$, $\F_{\rm D}$ be the unitary cosine and sine transforms on $\ltwo(\R_+)$, respectively.  The subscripts $\rm N$ and $\rm D$ are related to the Neumann Laplacian and the Dirichlet Laplacian in $\ltwo(\R_+)$, which are diagonalised by $\F_{\rm N}$ and $\F_{\rm D}$, respectively.  
Note also that these operators correspond to their own inverse.
It is then easily checked that
$$
\U \F_1\U^*=
\left(\begin{smallmatrix}
\F_{\rm N} & 0\\
0 & -i\F_{\rm D}
\end{smallmatrix}\right).
$$
In addition, by a straightforward computation one gets
$$
\U m(D)\U^* = \U \F_1^* m(X) \F_1 \U^*
= 
\left(\begin{smallmatrix}
\F_{\rm N} m_{\rm e}(L)\F_{\rm N} &  -i \F_{\rm N} m_{\rm o}(L)\F_{\rm D}\\
i \F_{\rm D}m_{\rm o}(L)\F_{\rm N} & \F_{\rm D} m_{\rm e}(L) \F_{\rm D}
\end{smallmatrix}\right).
$$

For the final step, let us introduce the unitary group of dilation in $\ltwo(\R)$
by $[U_tf](x)=\e^{t/2}f\big(\e^t x\big)$ for any $f\in \ltwo(\R)$, $x\in \R$, and $t\in \R$, whose generator is denoted by $A$. Similarly, the dilation group in $\ltwo(\R_+)$ is defined by
$[U_t^+ f](\ell)=\e^{t/2}f\big(\e^t \ell\big)$ for any $f\in \ltwo(\R_+)$, $\ell\in \R_+$, and $t\in \R$, with generator denoted by $A_+$.
Then one easily observes that
\begin{equation*}
\U m(A)\U^*=
\left(\begin{smallmatrix}
m(A_+) & 0\\
0 & m(A_+)
\end{smallmatrix}\right).
\end{equation*}
In addition, the following equality holds:  
$$
i\F_{\rm N} \F_{\rm D} = -\tanh(\pi A_+)+ i \cosh(\pi A_+)^{-1}=:\phi(A_+).
$$
We refer for example to \cite[Prop.~4.13]{DR} for a proof of the above equality.
We also stress that $|\phi|=1$, which will be used implicitly in the following computations.
Then, if we define the Neumann Laplacian by $-\Delta:=\F_{\rm N} L^2 \F_{\rm N}$,
we end up with 
\begin{equation}\label{eq:g3}
\U m(D)\U^* = \U \F_1^* m(X) \F_1 \U^*
= 
\left(\begin{smallmatrix}
m_{\rm e}\big(\sqrt{-\Delta}\big) &  - m_{\rm o}\big(\sqrt{-\Delta}\big)
\phi(A_+)\\
-\overline{\phi}(A_+) m_{\rm o}\big(\sqrt{-\Delta}\big) 
&\  \overline{\phi}(A_+) m_{\rm e}\big(\sqrt{-\Delta}\big) \phi(A_+)
\end{smallmatrix}\right).
\end{equation}

The next step is to look at the image of the operators $A_{N/2,N/2}$, $A_{N/2,N}$, $A_{N,N/2}$, and $A_{N,N}$ inside this new representation. The resulting expressions are rather long,
but the computation is easy, once we take into account that the commutators between these operators are compact. Indeed, it is shown in \cite[Sec.~5.5--5.7]{Cordes}
that any commutator between $a(A_+)$, $b(L)$, and $c(-\Delta)$ belongs to $\K\big(\ltwo(\R_+)\big)$  if $a\in C\big([-\infty,+\infty]\big)$ and $b, c \in C\big([0,+ \infty]\big)$.
Thus, we present the results modulo compact operators, which means
that the sign $\cong$ corresponds to an equality modulo $\K\big(\ltwo(\R_+;\C^2)\big)$.
One gets:
\begin{align}\label{eq_D1b}
\nonumber & \U \;\! A_{N/2,N/2}\;\! \U^*\\
\nonumber & \cong \left(\begin{smallmatrix}
1 & 0 \\
0 & 1
\end{smallmatrix}\right)+
\tfrac12\left(\begin{smallmatrix}
    1 & \tanh(\pi\DD)\phi(A_+)-i\sech(\pi\DD)\tanh(L)\\
    \overline{\phi}(A_+)\tanh(\pi\DD)-i\sech(\pi\DD)\tanh(L) & 1
\end{smallmatrix}\right)\\
& \quad \times \Big[\left(\begin{smallmatrix}
    \s^0_{N/2,N/2}(-2+2\tanh(\cdot))_{\rm e}(L) & \ \s^0_{N/2,N/2}(-2+2\tanh(\cdot))_{\rm o}(L)\\
    \s^0_{N/2,N/2}(-2+2\tanh(\cdot))_{\rm o}(L) &\ \s^0_{N/2,N/2}(-2+2\tanh(\cdot))_{\rm e}(L)
\end{smallmatrix}\right)-\left(\begin{smallmatrix}
1 & 0 \\
0 & 1
\end{smallmatrix}\right)\Big],
\end{align}
and 
\begin{align}\label{eq_D4b}
\nonumber & \U\;\!  A_{N,N}\;\!  \U^*\\
\nonumber & \cong \left(\begin{smallmatrix}
1 & 0 \\
0 & 1
\end{smallmatrix}\right)+
\tfrac12\left(\begin{smallmatrix}
    1 & \tanh(\pi\DD)\phi(A_+)-i\sech(\pi\DD)\tanh(L)\\
    \overline{\phi}(A_+)\tanh(\pi\DD)-i\sech(\pi\DD)\tanh(L) & 1
\end{smallmatrix}\right)\\
& \quad \times \Big[\left(\begin{smallmatrix}
    \s^0_{N,N}(2+2\tanh(\cdot))_{\rm e}(L) & \ \s^0_{N,N}(2+2\tanh(\cdot))_{\rm o}(L)\\
    \s^0_{N,N}(2+2\tanh(\cdot))_{\rm o}(L) &\ \s^0_{N,N}(2+2\tanh(\cdot))_{\rm e}(L)
\end{smallmatrix}\right)-\left(\begin{smallmatrix}
1 & 0 \\
0 & 1
\end{smallmatrix}\right)\Big].
\end{align}
Similarly, one has
\begin{align}\label{eq_D2b}
\nonumber & \U \;\! A_{N/2,N}\;\! \U^*\\
& \cong \left(\begin{smallmatrix}1 & 0 \\ 0 & -1\end{smallmatrix}\right)
\left(\begin{smallmatrix}
\overline{\psi}_{\rm e}\big(\sqrt{-\Delta}\big) &  - \overline{\psi}_{\rm o}\big(\sqrt{-\Delta}\big)
\phi(A_+)\\
-\overline{\phi}(A_+) \overline{\psi}_{\rm o}\big(\sqrt{-\Delta}\big) 
&\  \overline{\psi}_{\rm e}\big(\sqrt{-\Delta}\big)
\end{smallmatrix}\right)
\left(\begin{smallmatrix}
(m_{N/2,N})_{\rm e}(L)~ &  (m_{N/2,N})_{\rm o}(L)\\
(m_{N/2,N})_{\rm o}(L)~ &  (m_{N/2,N})_{\rm e}(L)
\end{smallmatrix}\right)
\end{align}
and
\begin{align}\label{eq_D3b}
\nonumber & \U \;\! A_{N,N/2}\;\! \U^*\\
& \cong -\left(\begin{smallmatrix}
\overline{\psi}_{\rm e}\big(\sqrt{-\Delta}\big) &  - \overline{\psi}_{\rm o}\big(\sqrt{-\Delta}\big)
\phi(A_+)\\
-\overline{\phi}(A_+) \overline{\psi}_{\rm o}\big(\sqrt{-\Delta}\big) 
&\  \overline{\psi}_{\rm e}\big(\sqrt{-\Delta}\big)
\end{smallmatrix}\right)
\left(\begin{smallmatrix}
(m_{N,N/2})_{\rm e}(L)~ &  (m_{N,N/2})_{\rm o}(L)\\
(m_{N,N/2})_{\rm o}(L)~ &  (m_{N,N/2})_{\rm e}(L)
\end{smallmatrix}\right)
 \left(\begin{smallmatrix}1 & 0 \\ 0 & -1\end{smallmatrix}\right).
\end{align}

Let us now recall a construction of Cordes. 
In \cite[Sec.~V.7]{Cordes}, the following $C^*$-subalgebra
of $\B\big(L^2(\R_+)\big)$ is introduced: 
$$
\EEC:=C^*\Big(a_i(A_+)b_i(L)c_i(-\Delta)\mid a_i\in C\big([-\infty,+\infty]\big), \ b_i,c_i\in C\big([0,+ \infty]\big)\Big).
$$ 
It is then shown in \cite[Thm.~V.7.3]{Cordes} that the quotient algebra $\EEC/\K\big(L^2(\R_+)\big)$ is isomorphic to $C( \hexagon)$, the set of continuous functions defined on the edges of a hexagon
see Figure \ref{fig1}. 
For an operator of the form  
$a(A_+)\;\!b(L)\;\!c\big(\sqrt{-\Delta}\big) \in \EEC$, 
its image in the quotient algebra takes the form
\begin{align}\label{eq:6c}
\begin{split}
&\Gamma^{1}(\ell):=a(-\infty)\;\!b(\ell)\;\!c(+\infty), \qquad \ell \in [0,+\infty], \\
&\Gamma^{2}(\xi):=a(-\infty)\;\!b(+\infty)\;\!c(\xi), \qquad \xi \in [+\infty,0], \\
&\Gamma^{3}(s):=a(s)\;\!b(+\infty)\;\!c(0), \qquad s\in [-\infty, +\infty], \\
&\Gamma^{4}(\xi):=a(+\infty)\;\!b(+\infty)\;\!c(\xi), \qquad \xi \in [0,+\infty], \\
&\Gamma^{5}(\ell):=a(+\infty)\;\!b(\ell)\;\!c(+\infty), \qquad \ell\in [+\infty,0], \\
&\Gamma^{6}(s):=a(s)\;\!b(0)\;\!c(+\infty), \qquad s\in [+\infty, -\infty].
\end{split}
\end{align}
Observe that we gave an orientation on the interval on which these functions are
defined. As a result, the concatenation map
$$
\Gamma\equiv (\Gamma^{1},\Gamma^{2},\Gamma^{3},\Gamma^{4},\Gamma^{5},\Gamma^{6}): \hexagon \to \C
$$
is continuous, even at the vertices of the hexagon.

\begin{figure}[!ht]
\centering
\tdplotsetmaincoords{80}{0}
    \begin{tikzpicture}[scale = 4, tdplot_main_coords]
        \coordinate (O) at (0,0,0);
        
        \coordinate (A1) at (1,0,0);
        \coordinate (A2) at (0.4,1.5,0);
        \coordinate (A3) at (1.4,1.5,0);
        \coordinate (B0) at (0,0,1);
        \coordinate (B1) at (1,0,1);
        \coordinate (B2) at (0.4,1.5,1);
        \coordinate (B3) at (1.4,1.5,1);

        \draw[{Triangle[length=2mm, width=2mm]}-, line width = 0.75mm] (A1) node[inner sep = 1.5mm, anchor=north east]{$\infty$} to node [below]{$\ell$} (O) node[inner sep = 1mm, anchor=120]{$0$};
        \draw[dashed, gray] (O) -- (A2) ; %
        \draw[{Triangle[length=2mm, width=2mm]}-,line width = 0.75mm] (A3) node[inner sep = 2.0mm, anchor=80]{$0$} to node [below]{$\xi$} (A1) node[inner sep = 0.5mm, anchor=north west]{$\infty$};
        \draw[gray, dashed](A2) -- (A3); %
        \draw[{Triangle[length=2mm, width=2mm]}-,line width = 0.75mm] (B0) node[inner sep = 1mm, anchor=250]{$0$} to node [above]{$\ell$} (B1) node[anchor=south east]{$\infty$};
        \draw[gray] (B0) -- (B2);  %
        \draw[{Triangle[length=2mm, width=2mm]}-, line width = 0.75mm] (B1) node[inner sep = 2.5mm, anchor=260]{$\infty$} to node [above]{$\xi$} (B3) node[inner sep = 0.5mm, anchor=300]{$0$};
        \draw[gray] (B2) -- (B3);  %
        \draw[{Triangle[length=2mm, width=2mm]}-,line width = 0.75mm] (O) node[inner sep = 1.0mm, anchor=-20]{$-\infty$} to node [left]{$s$} (B0) node[inner sep = 1mm, anchor=30]{$\infty$};
        \draw[gray] (A1) -- (B1);  %
        \draw[gray,dashed] (A2) -- (B2);  %
        \draw[{Triangle[length=2mm, width=2mm]}-,line width = 0.75mm] (B3) node[inner sep = 1mm, anchor=150]{$\infty$} to node [right]{$s$} (A3) node[inner sep = 0.75mm, anchor=200]{$-\infty$};
    \end{tikzpicture}
\caption{Representation of the quotient algebra, with orientation indicated on the edges. The starting point of $\Gamma^{1}$ is located on the lower left corner.}
\label{fig1}
\end{figure}

Our interest in the construction of Cordes comes from the similarity between the elements of $\EEC$
and the expressions provided in \eqref{eq_D1b}, \eqref{eq_D4b}, \eqref{eq_D2b}, and \eqref{eq_D3b}. Indeed, the functions of the three operators $L$, $A_+$, $-\Delta$ are continuous, and have limits either at $-\infty$ and $+\infty$, or at $0$ and $+\infty$. 
However, in order to accommodate all operators in a single algebra, we have to consider
the unital $C^*$-algebra $M_4(\EEC)$, the $4\times 4$ matrices with values in $\EEC$. Clearly, this algebra contains the ideal
$M_4\big(\K(\ltwo(\R_+))\big)$, and one has
\begin{equation}\label{eq_quotient}
M_4(\EEC) \big/ M_4\big(\K(\ltwo(\R_+))\big)
= M_4\big(C(\hexagon)\big)\cong C\big(\hexagon; M_4(\C)\big)
\end{equation}
One can thus look at the image of the operators obtained above through the quotient map
$$
q: M_4(\EEC) \to C\big(\hexagon; M_4(\C)\big)
$$
with kernel $M_4\big(\K(\ltwo(\R_+))\big)$.
In the next statement we provide this image, keeping the convention provided in \eqref{eq:6c}
for the enumeration of the $6$ components.

Since some of the expressions obtained through the quotient map are rather long, let us use a block matrix notation. For example, the expression $\Gamma^j$, which is a $4\times 4$-matrix-valued function will be decomposed into four $2\times 2$-matrix-valued functions according to the notation
\begin{equation}\label{eq_Gamma^j}
\Gamma^{j} = \left(\begin{smallmatrix}
\Gamma^{j}_{N/2,N/2} & \Gamma^{j}_{N/2,N} \\
\Gamma^{j}_{N,N/2} & \Gamma^{j}_{N,N}  
\end{smallmatrix}\right).
\end{equation}
Note that the blocks are in direct relation with the notation introduced in \eqref{eq_def_block}.
By remembering that for any $f:\R\to C$ one has $f_{\rm e}(x)+f_{\rm o}(x)=f(x)$ and $f_{\rm e}(x)-f_{\rm o}(x)=f(-x)$, we readily obtain from \eqref{eq_D1b} and for $\ell\in \R_+$ that
\begin{align*}
\Gamma^1_{N/2,N/2}(\ell) & =\tfrac12\left(\begin{smallmatrix}
    1 & -1\\
    -1 & 1
\end{smallmatrix}\right)+\tfrac12 \s^0_{N/2,N/2}(-2+2\tanh(\ell))\left(\begin{smallmatrix}
    1 & 1\\
    1 & 1
\end{smallmatrix}\right), \\
\Gamma^5_{N/2,N/2}(\ell) & = \tfrac12\left(\begin{smallmatrix}
    1 & 1\\
    1 & 1
\end{smallmatrix}\right)+\tfrac12 \s^0_{N/2,N/2}(-2+2\tanh(-\ell))\left(\begin{smallmatrix}
    1 & -1\\
    -1 & 1
\end{smallmatrix}\right).
\end{align*} 
Recalling that $f_{\rm o}(0)=0$ and for $s\in \R$ we also obtain
\begin{equation*}
\Gamma^6_{N/2,N/2}(s) = \tfrac12\left(\begin{smallmatrix}
    1+  \s^0_{N/2,N/2}(-2)  & \ \phi(s)( \s^0_{N/2,N/2}(-2) -1)\\
    \overline{\phi}(s)( \s^0_{N/2,N/2}(-2) -1)  &\   1+  \s^0_{N/2,N/2}(-2) 
\end{smallmatrix}\right).
\end{equation*}
For the remaining three functions, let us define for $\xi \in \R_+$ 
\begin{equation*}
\eta_\pm(\xi):=\tanh(\pi\xi)\pm i\sech(\pi\xi)
\end{equation*}
and 
$$
s_{N/2}^{\pm}:=\frac{\s^0_{N/2,N/2}(0_-)\pm \s^0_{N/2,N/2}(-4)}2.
$$ 
With these notations we get
\begin{align*}
\Gamma^2_{N/2,N/2} (\xi)
& =\tfrac12\left(\begin{smallmatrix}
    1+s_{N/2}^++\eta_-(\xi)s_{N/2}^- & -\eta_-(\xi)+s_{N/2}^-+\eta_-(\xi)s_{N/2}^+\\
    -\eta_-(\xi) +\eta_-(\xi)s_{N/2}^++s_{N/2}^- & 1+ \eta_-(\xi)s_{N/2}^-+s_{N/2}^+
\end{smallmatrix}\right), \\
\Gamma^3_{N/2,N/2}(s) &  =\tfrac12\left(\begin{smallmatrix}
    1+s_{N/2}^+-is_{N/2}^- & i+s_{N/2}^--is_{N/2}^+\\
    i -is_{N/2}^++s_{N/2}^- & 1- is_{N/2}^-+s_{N/2}^+
\end{smallmatrix}\right), \\
\Gamma^4_{N/2,N/2}(\xi) & =\tfrac12\left(\begin{smallmatrix}
    1+s_{N/2}^+-\eta_+(\xi)s_{N/2}^- & \eta_+(\xi)+s_{N/2}^--\eta_+(\xi)s_{N/2}^+\\
    \eta_+(\xi) -\eta_+(\xi)s_{N/2}^++s_{N/2}^- & 1- \eta_+(\xi)s_{N/2}^-+s_{N/2}^+
\end{smallmatrix}\right).
\end{align*}

Similarly, starting from \eqref{eq_D4b} one infers that
\begin{align*}
\Gamma^1_{N,N}(\ell) & =\tfrac12\left(\begin{smallmatrix}
    1 & -1\\
    -1 & 1
\end{smallmatrix}\right)+\tfrac12 \s^0_{N,N}(2+2\tanh(\ell))\left(\begin{smallmatrix}
    1 & 1\\
    1 & 1
\end{smallmatrix}\right), \\
\Gamma^5_{N,N}(\ell) & =\tfrac12\left(\begin{smallmatrix}
    1 & 1\\
    1 & 1
\end{smallmatrix}\right)+\tfrac12 \s^0_{N,N}(2+2\tanh(-\ell))\left(\begin{smallmatrix}
    1 & -1\\
    -1 & 1
\end{smallmatrix}\right), 
\end{align*}
and also
\begin{equation*}
\Gamma^6_{N,N}(s)  = \tfrac12\left(\begin{smallmatrix}
    1+  \s^0_{N,N}(2)  & \ \phi(s)( \s^0_{N,N}(2) -1)\\
    \overline{\phi}(s)( \s^0_{N,N}(2) -1)  &\   1+  \s^0_{N,N}(2) 
\end{smallmatrix}\right).
\end{equation*}
Then, by setting 
$$
s_N^\pm:=\frac{\s^0_{N,N}(4)\pm \s^0_{N,N}(0_+)}{2}
$$ 
we infer that 
\begin{align*}
\Gamma^2_{N,N}(\xi)
& =\tfrac12\left(\begin{smallmatrix}
    1+s_{N}^++\eta_-(\xi)s_{N}^- & -\eta_-(\xi)+s_{N}^-+\eta_-(\xi)s_{N}^+\\
    -\eta_-(\xi) +\eta_-(\xi)s_{N}^++s_{N}^- & 1+ \eta_-(\xi)s_{N}^-+s_{N}^+
\end{smallmatrix}\right), \\
\Gamma^3_{N,N}(s)
& =\tfrac12\left(\begin{smallmatrix}
    1+s_{N}^+-is_{N}^- & i+s_{N}^--is_{N}^+\\
    i -is_{N}^++s_{N}^- & 1- is_{N}^-+s_{N}^+
\end{smallmatrix}\right), \\
\Gamma^4_{N,N}(\xi)
& =\tfrac12\left(\begin{smallmatrix}
    1+s_{N}^+-\eta_+(\xi)s_{N}^- & \eta_+(\xi)+s_{N}^--\eta_+(\xi)s_{N}^+\\
    \eta_+(\xi) -\eta_+(\xi)s_{N}^++s_{N}^- & 1- \eta_+(\xi)s_{N}^-+s_{N}^+
\end{smallmatrix}\right).
\end{align*}

Starting from \eqref{eq_D2b} one also gets 
\begin{align*}
\Gamma^1_{N/2,N} & = \Gamma^5_{N/2,N} = \Gamma^6_{N/2,N}= \left(\begin{smallmatrix} 0 & 0 \\ 0 & 0 \end{smallmatrix}\right), \\
\Gamma^2_{N/2,N}(\xi) &
= \frac{1}{2\sqrt{2\pi}}(1+i)\alpha \overline{\psi}(\xi) \left(\begin{smallmatrix} -1 & 1 \\ -1 & 1 \end{smallmatrix}\right), \\
\Gamma^3_{N/2,N}(s) & 
= \frac{1}{2\sqrt{2\pi}}(1+i)\alpha \overline{\psi}(0) \left(\begin{smallmatrix} -1 & 1 \\ -1 & 1 \end{smallmatrix}\right), \\
\Gamma^4_{N/2,N}(\xi) &
= \frac{1}{2\sqrt{2\pi}}(1+i)\alpha \overline{\psi}(-\xi) \left(\begin{smallmatrix} -1 & 1 \\ -1 & 1 \end{smallmatrix}\right), 
\end{align*}
and similarly from \eqref{eq_D3b}
\begin{align*}
\Gamma^1_{N,N/2} & = \Gamma^5_{N,N/2} = \Gamma^6_{N,N/2} = \left(\begin{smallmatrix} 0 & 0 \\ 0 & 0 \end{smallmatrix}\right), \\
\Gamma^2_{N,N/2}(\xi) & 
= \frac{1}{2\sqrt{2\pi}}(1-i)\alpha \overline{\psi}(\xi) \left(\begin{smallmatrix} -1 & -1 \\ 1 & 1 \end{smallmatrix}\right), \\
\Gamma^3_{N,N/2}(s) & 
= \frac{1}{2\sqrt{2\pi}}(1-i)\alpha \overline{\psi}(0) \left(\begin{smallmatrix} -1 & -1 \\ 1 & 1 \end{smallmatrix}\right),  \\
\Gamma^4_{N,N/2}(\xi) & 
= \frac{1}{2\sqrt{2\pi}}(1-i)\alpha \overline{\psi}(-\xi) \left(\begin{smallmatrix} -1 & -1 \\ 1 & 1 \end{smallmatrix}\right).
\end{align*}

Let us still compute the pointwise determinant of the matrix-valued function
$$
\Gamma\equiv (\Gamma^{1},\Gamma^{2},\Gamma^{3},\Gamma^{4},\Gamma^{5},\Gamma^{6}): \hexagon \to M_4(\C)
$$
with $\Gamma^j$ defined in \eqref{eq_Gamma^j}.
Since $\Gamma^j_{N/2,N}=0$ and $\Gamma^j_{N,N/2}=0$ for $j\in \{1, 5, 6\}$, the matrices
$\Gamma^1$,  $\Gamma^5$, and $\Gamma^6$ are block diagonal, with the their diagonal
elements given by $\Gamma^j_{N/2,N/2}$ and $\Gamma^j_{N,N}$. In this case, the computation
of the pointwise determinant is easy and one gets
\begin{align*}
\det (\Gamma^1)(\ell) & = \s^0_{N/2,N/2}\big(-2+2\tanh(\ell)\big)\;\!\s^0_{N,N}\big(2+2\tanh(\ell)\big), \\
\det (\Gamma^5)(\ell) & =\s^0_{N/2,N/2}\big(-2+2\tanh(-\ell)\big)\;\!\s^0_{N,N}\big(2+2\tanh(-\ell)\big), \\
\det (\Gamma^6)(s) & = \s^0_{N/2,N/2}(-2)\;\!\s^0_{N,N}(2).
\end{align*}

For the remaining three parts, namely for $\Gamma^2$, $\Gamma^3$, and $\Gamma^4$, 
these matrices are no more block diagonal, but nevertheless they share a specific form
which is analyzed in the next statement.
Its easy proof is not provided.

\begin{Lemma}\label{lemma_abc}
For $a,b,c,d,e,f\in \C$, one has
    \begin{equation}\label{eq_4mat}
        \det\begin{pmatrix}
a&b&-c&c\\
b&a&-c&c\\
-d&-d&e&f\\
d&d&f&e
        \end{pmatrix}=(b-a)(f+e)\big(4cd+(b+a)(f-e)\big)\ .
    \end{equation}
    In particular, if $c=d=0$ we obtain $(a^2-b^2)(e^2-f^2)$.
\end{Lemma}

The following result will also be used, and can be obtained with standard relations 
for hyperbolic trigonometric functions:

\begin{Lemma}
For the function $\psi:\R\to \C$ defined on $y\in \R$ by
$$
\psi(y):=\sqrt{\pi}\ \frac{\cosh(\pi y/2)-i\sinh(\pi y/2)}{\cosh(\pi y)}.
$$
one has for $\xi\geq 0$: 
\begin{equation}
\overline{\psi}(\pm\xi)^2=\pm i\pi \sech(\pi \xi)\;\!\eta_{\mp}(\xi) ,
\end{equation}
where $\eta_\pm$ has been defined in \eqref{eq_def_eta}.
\end{Lemma}

Let us firstly focus on $\Gamma^4$ and keep the notation of Lemma \ref{lemma_abc}.
For this matrix-valued function one gets
\begin{align*}
    a&=\tfrac{1}{2}\big(1+s_{N/2}^+-\eta_+(\xi)s_{N/2}^-\big),\\
    b&=\tfrac{1}{2}\big(\eta_+(\xi) -\eta_+(\xi)s_{N/2}^++s_{N/2}^-\big),\\
    c&=\tfrac{1}{2\sqrt{2\pi}}(1+i)\alpha \overline{\psi}(-\xi), \\
    d&=\tfrac{1}{2\sqrt{2\pi}}(1-i)\alpha \overline{\psi}(-\xi),\\
    e&=\tfrac{1}{2}\big(1+s_{N}^+-\eta_+(\xi)s_{N}^-\big),\\
    f&=\tfrac{1}{2}\big(\eta_+(\xi) -\eta_+(\xi)s_{N}^++s_{N}^-\big),
\end{align*}
leading to 
\begin{align*}
   b-a&=\tfrac{1}{2}\big(1-\s^0_{N/2,N/2}(-4)\big)\eta_+(\xi) -\tfrac{1}{2}\big(\s^0_{N/2,N/2}(-4)+1\big),\\
   b+a&=\tfrac{1}{2}\big(1-\s^0_{N/2,N/2}(0_-)\big)\eta_+(\xi) +\tfrac{1}{2}\big(\s^0_{N/2,N/2}(0_-)+1\big),\\
   4cd&=\tfrac{\alpha^2}{\pi}\overline{\psi}(-\xi)^2=-\alpha^2 i \sech(\pi \xi)\eta_+(\xi),\\
   f-e&=\tfrac{1}{2}\big(1-\s^0_{N,N}(0_+)\big)\eta_+(\xi) -\tfrac{1}{2}\big(\s^0_{N,N}(0_+)+1\big),\\
   f+e&=\tfrac{1}{2}\big(1-\s^0_{N,N}(4)\big)\eta_+(\xi) +\tfrac{1}{2}\big(\s^0_{N,N}(4)+1\big).
\end{align*}
With these expressions at hand, the function $\xi\mapsto \Gamma^4(\xi)$
can be obtained by the r.h.s.~of \eqref{eq_4mat}.
Also, a direct inspection shows that for $\Gamma^2$, the results are similar, with $\eta_+$
to be replaced by $-\eta_-$. Similarly, $\Gamma^3$ can be obtained from the above expressions
by replacing $\eta_+(\xi)$ by $i$, which corresponds to $\eta_+(0)$.
The even more explicit expressions for $N=2$ will be provided in the next section.

\subsection{Topological Levinson's theorem for N=2}\label{sec_top_N2}

The interest in the special $N=2$ comes from the fact that $\Hrond^0_{\rm int}$, 
as introduced in \eqref{eq_def_int}, can naturally be identified with $\Hrond^0$.
Accordingly, the content of Proposition \ref{prop_restriction} does not correspond to
a restriction of the wave operator to a subspace, but it corresponds to the full wave operator.
Thus, after the unitary conjugation through $\U$ introduced in the previous section,
we are still dealing with a unitarily equivalent representation of the wave operator $W_-^0$
and not to any restriction of this operator.

We firstly provide the information about the scattering operator at $\lambda=0$.

\begin{Lemma}\label{lem_aa}
When $\v\xi^0_1$ and $\v\xi^0_2$ are linearly independent, 
 $\s^0_{1,1}(0_-) = -1$ and $\s^0_{2,2}(0_+) =  -1$. 
When there exists $\alpha\in \{-1,1\}$ satisfying $\v\xi^0_{1}=\alpha \v \xi^0_2$,
one has $\s^0_{1,1}(0_-) = -i$ and $\s^0_{2,2}(0_+) =  i$.
\end{Lemma}

Before the proof, let us recall that the precise expression for $S^\theta(\lambda)$
has been provided in \cite[Thm.~3.9]{NRT} for $\lambda$ belonging to all thresholds in $\T^\theta$.
Here, we only need the special cases $\lambda \in \{-4,0,4\}$ obtained for $\theta=0$ and for 
$N=2$. The notations are
similar to those already used in the proof of Lemma \ref{lem_limit1}.

\begin{proof}
We are dealing with the threshold $\lambda=0$.
Recall from \cite[Sec.~4.3]{NRT} that $I_0(0)=\frac{1}{2}\big(\v \P^0_2\v+i\v \P^0_{1}\v\big)$.
In the non-intricate case, namely when $\v \xi^0_1$ and $\v \xi^0_2$ are linearly independent
one has
\begin{equation*}
\P^0_j \v \Big(\tfrac{1}{2}\big(\v \P^0_2\v+i\v \P^0_{1}\v\big)\Big)^{-1} \v \P_j
=2 \big|\xi^0_j\big\rangle  \big\langle \v \xi^0_j\big|
\Big(|\v\xi^0_2\rangle  \langle \v \xi^0_2|+ i(|\v\xi^0_1\rangle  \langle \v \xi^0_1|\Big)^{-1}
\big|\v\xi^0_j\big\rangle  \big\langle \xi^0_j\big|
\end{equation*}
which is equal to $2\P^0_2$ in the case $j=2$ and to $-2i\P^0_1$ in the case $j=1$,
by \cite[Lem.~A.1]{ANRR}.
By the explicit expressions for the scattering matrix provided in \cite[Thm.~3.9]{NRT}, 
we infer that
$S^0(0_-)_{1,1}=-\P^0_{1}$ and $S^0(0_+)_{2,2}=-\P^0_{2}$.

If $\v \xi^0_1=\alpha \v \xi^0_2$ with $\alpha\in \{-1,1\}$, as shown in the proof of Lemma \ref{lem_limit1}
we have 
$I_0(0)=\frac{1}{4} \|v\|_{\ltwo}^2 (1+i)  \big|\widehat{\v\xi_{2}^0}\big\rangle \big\langle \widehat{\v\xi_{2}^0}\big|$. It then follows that
$$
\P^0_j \v I_0(0)^{-1}\v\P^0_j = (1-i)\P^0_j.
$$
In addition, the next term is the expansion of the scattering operator 
at threshold $0$, since the projection $S_1$ is $0$,  
\cite[Lem.~4.9]{NRT}.
By the explicit formula for the scattering matrix we finally get
$S^0(0_-)_{1,1}=-i\P^0_{1}$ and $S^0(0_+)_{2,2}=i\P^0_{2}$.
\end{proof}

For the thresholds $-4$ and $4$, it has been shown in \cite[Lem.~4.3]{APRR} that $\s^0_{1,1}(-4)$ and $\s^0_{2,2}(4)$
are generically equal to $-1$, and are equal to $1$ when a resonance at $\lambda=-4$ or $\lambda=4$ takes place. In the next statement, we show that in the intricate case, only 
the values $-1$ can appear.

\begin{Lemma}\label{lem_bb}
If there exists $\alpha\in \{-1,1\}$ satisfying $\v\xi^0_{1}=\alpha \v \xi^0_2$,
then $\s^0_{1,1}(-4)=-1$ and $\s^0_{2,2}(4)=-1$.
\end{Lemma}
 
\begin{proof}
Observe that the condition of the statement is satisfied if and only if $v(1)=0$ and $v(2)\neq 0$ or
equivalently if $v(1)\neq 0$ and $v(2)=0$. For simplicity, let us denote by $a\neq 0$ the non-zero value of the potential, while its other value is $0$

We provide the proof for $\s_{1,1}^0(-4)$, the proof for $\s^0_{2,2}(4)$ being similar.
The first step is to recall the formula provided in \cite[Thm.~3.9.(b)]{NRT}, namely
\begin{equation}\label{eq:S11}
S^0(-4)_{1,1}
=\P_1^0-\P_1^0\v\big(I_0(0)+S_0\big)^{-1}\v\;\!\P_{1}^0
+\P_1^0\v\;\!C_{10}'(0)S_1\big(I_2(0)+S_2\big)^{-1}S_1C_{10}'(0)\v\;\!
\P_{1}^0.
\end{equation}
In this case we have $I_0(0)=\frac{1}{4}|\v\xi_1^0\rangle \langle \v\xi_1^0|$, and then
$\P_1^0\v(I_0(0)+S_0)^{-1}\v\P_1^0=2\P_1^0$. Thus, if the  third term in the r.h.s.~of \eqref{eq:S11} vanishes, then $\s_{1,1}^0(-4)=-1$ holds.

Let us recall from the proof of \cite[Prop.~3.5]{NRT} that $S_1$ is defined as the projection on the kernel of $I_1(0)$ (inside the subspace $S_0\C^2$). Thus, in order to deal with the third term, we have to study $I_1(0)$, and a few simple computations lead to the equalities:
\begin{equation*}
I_1(0) 
= S_0 M_1(0) S_0 
= S_0\u S_0+S_0\frac{|\v \xi_2^0\rangle \langle \v \xi_2^0|}{8\sqrt{2}}S_0
=S_0,
\end{equation*}
from which we infer that $S_1=0$, and the triviality of the third term.  
\end{proof}

With the results obtained above, we can now easily provide the expressions for $\det(\Gamma^2)(\xi)$, for 
$\det(\Gamma^3)(s)$, and for $\det(\Gamma^4)(\xi)$ according to the behaviors at the thresholds
$-4$, $0$, and $4$. Observe that the non-intricate case can be deduced from the general formulas of the previous section simply by setting $\alpha =0$. By collecting the various results, we get:
\begin{enumerate}
\item[$(i)$] 
In the non-intricate case ($\alpha=0$) and generic case, namely for $\s^0_{1,1}(-4)=\s^0_{1,1}(0_-)=\s^0_{2,2}(0_+)=\s^0_{2,2}(4)=-1$, we have
$$ 
\det (\Gamma^2)(\xi)=\eta_-(\xi)^4, \quad \det(\Gamma^3)(s)=1, \quad\det (\Gamma^4)(\xi)=\eta_+(\xi)^4,
$$
\item[$(ii)$] In the non-intricate case ($\alpha=0$) and left resonant case, namely for $\s^0_{1,1}(-4)=1$ and $\s^0_{1,1}(0_-)=\s^0_{2,2}(0_+)=\s^0_{2,2}(4)=-1$ we have
$$
\det (\Gamma^2)(\xi)=\eta_-(\xi)^3, \quad \det(\Gamma^3)(s)=i, \quad \det (\Gamma^4)(\xi)=-\eta_+(\xi)^3,
$$
\item[$(iii)$]
In the non-intricate case ($\alpha=0$) and right resonant case, namely for $\s^0_{2,2}(4)=1$ and $\s^0_{1,1}(0_-)=\s^0_{2,2}(0_+)=\s^0_{1,1}(-4)=-1$ we have
$$
\det (\Gamma^2)(\xi)=-\eta_-(\xi)^3, \quad \det(\Gamma^3)(s)=-i, \quad \det (\Gamma^4)(\xi)=\eta_+(\xi)^3,
$$
\item[$(iv)$]
In the intricate case ($\alpha^2=1$) and for $\s^0_{1,1}(-4)=\s^0_{2,2}(4)=-1$, $\s^0_{1,1}(0_-)=-i$, $\s^0_{2,2}(0_+)=i$ we have
$$
\det (\Gamma^2)(\xi)=i\eta_-(\xi)^3, \quad \det(\Gamma^3)(s)=-1, \quad \det (\Gamma^4)(\xi)=-i\eta_+(\xi)^3.
$$
\end{enumerate}

Based on the expressions obtained above
we provide here the statement of Levinson's theorem in the 
special case $N=2$. The general case will be presented in the next section.

As shown above and in the special case $N=2$,  the unital $C^*$-algebra  
$M_4(\EEC)$ contains the wave operator $W_-^0$, once suitable unitary conjugations are applied.
In addition, this algebra contains the ideal $M_4\big(\K(\ltwo(\R_+))\big)$.
Then, as a consequence of Cordes' result, namely \eqref{eq_quotient}, one has the short exact sequence of $C^*$-algebras
$$
0 \longrightarrow M_4\big(\K(\ltwo(\R_+))\big) \longrightarrow
M_4(\EEC) \stackrel{q}{\longrightarrow} 
C\big(\hexagon; M_4(\C)\big) \longrightarrow 0
$$
and the corresponding $6$ terms exact sequence for the $K$-theory of these algebras.
In particular, it is well known that $K_0\big(M_4\big(\K(\ltwo(\R_+))\big)\big)\cong\Z$
and that $K_1\big(C\big(\hexagon; M_4(\C)\big)\big)\cong\Z$.

Since the matrix-valued function $\Gamma= (\Gamma^1,\Gamma^2,\Gamma^3,\Gamma^4,\Gamma^5,\Gamma^6)$ exhibited in the previous subsection
belongs to $C\big(\hexagon; M_4(\C)\big)$ and is invertible, it defines an element
$[\Gamma]_1$ in the $K_1$-group of this algebra.
In addition, since $\W_-^0:=(\U\V^0\F^0)W_-^0(\U\V^0\F^0)^*\in M_4(\EEC)$ is an isometry and a lift for $\Gamma$, one directly infers from \cite[Prop.~9.2.4.(ii)]{RLL} that
\begin{equation}\label{eq:Lev1*}
\ind\big([\Gamma]_1\big) = [1-(\W^0_-)^*\W_-^0]_0-[1-\W_-^0(\W_-^0)^*]_0 =-[\E_{\rm p}(H^0)]_0,
\end{equation}
with $\E_{\rm p}(H^0)$ the image of the projection on the subspace spanned by the eigenfunctions of $H^0$ through the unitary conjugation by the unit defined by 
$\U\V^0\F^0$.
Let us emphasize that the equality \eqref{eq:Lev1*} corresponds to the topological version of Levinson's theorem: it is a relation (by the index map) between the equivalence class in $K_1$
of quantities related to scattering theory, and the equivalence class in $K_0$ of the projection
on the bound states of $H^0$. 

The standard formulation of Levinson's theorem is an equality between numbers. 
Thus, our last task is to extract a numerical equality from \eqref{eq:Lev1*}. In a more general setting we might pair the $K_1$ class of the scattering matrix with the Chern character of a suitable spectral triple, as in \cite{AR23}. For this specific case, we proceed in a more elementary way by using the determinant and winding number directly.  
Thus, on $M_4\big(\K(\ltwo(\R_+))\big)$, one uses the usual trace (on finite dimensional projections), and on 
$C\big(\hexagon; M_4(\C)\big)$
the winding number of the pointwise determinant is the correct notion to be used.

\begin{Remark}\label{rem:conv}
When computing the winding number, and pairing the equality \eqref{eq:Lev1*}
with traces, a few conventions about signs have to be taken. In the present setting, we shall turn around the hexagon anti-clockwise, 
and the increase in the winding number is counted clockwise.
The convention about the path is illustrated in Figure \ref{fig1}, 
with the starting point of \  $\Gamma^1$ located on the lower left corner.
\end{Remark}

The next statement summarized our findings for $N=2$ and $\theta=0$. In its statement, 
the numbers $(i)$ to $(iv)$ refer to the four cases mentioned earlier in this section.

\begin{Theorem}\label{thm_N=2}
For $N=2$ the following equality holds:
\begin{equation}
\# \sigma_{\rm p}(H^0)= \Var \big(\lambda \to \det S^0(\lambda)\big) +
\begin{cases} 2 & \hbox{ in case (i),} \\
3/2 & \hbox{ in case (ii)--(iv),}
\end{cases}
\end{equation}
where $ \Var \big(\lambda \to \det S^0(\lambda)\big)$
is the total variation
of the argument of the piecewise continuous function $\lambda \to \det S^0(\lambda)$.
\end{Theorem}

Note that the r.h.s.~can be further decomposed into
$$
\Var \big(\lambda \to \det S^0(\lambda)\big)
=\Var \big(\lambda \to \s_{1,1}^0(\lambda)\big)
+  \Var \big(\lambda \to \s_{2,2}^0(\lambda)\big).
$$
For the straightforward proof of this statement, we also use that 
$$
\Var \big(\lambda \to  [fg](\lambda)\big)=\Var \big(\lambda \to  f(\lambda)\big)
+ \Var \big(\lambda \to  g(\lambda)\big)
$$
for $f$ and $g$ two functions taking values in the set of complex numbers of modulus $1$.

\begin{Remark}\label{rem_N=2}
In the previous theorem, only the case $(iv)$ is new and corresponds to the intricate case. The other three cases can also be deduced from Theorem \ref{thm:Adam}
when $N=2$. For the record and anticipating the general case presented in the next section, let us rewrite the statement only for the intricate case. In the intricate case and for $N=2$ one has
\begin{equation}
\Var \big(\lambda \to \det S^0(\lambda)\big) + 2 -
\frac{\#\big\{k\mid \tilde \s^0_{k}(\lambda^0_k\pm 2)=1 \big\}}{2} -\frac{1}{2}
=\# \sigma_{\rm p}(H^0).
\end{equation}
\end{Remark}

\subsection{Topological Levinson's theorem for arbitrary N}

In this section we provide an extension of the framework introduced in Section \ref{sec_top_N2} for the intricate case and for $N=2$. It means that we shall still consider $\theta=0$ but
arbitrary $N\in 2\N$.
As already mentioned, the algebra introduced in \eqref{eq_def_EE} is not large enough in the intricate case since it can not accommodate the off-diagonal terms exhibited in Section \ref{sec_extra_t}.
On the other hand, the algebra introduced in Section \ref{sec_Cordes_alg} is suitable
only for $N=2$. Thus, a new construction is necessary, but 
fortunately the enlarged algebra  can be easily understood thanks to the study performed in
the previous sections.

Recall that in the intricate case, all complications are coming from the operators described  in \eqref{eq_MA12} and in \eqref{eq_MA21}. Indeed, these operators 
are not compact. However, note that this lack of compactness appears only in the two terms recalled in \eqref{eq_anno1} and in \eqref{eq_anno2}, independently of 
$N\in 2\N$.
More precisely, considering again \eqref{eq_formula}, \eqref{eq_MA12} and  \eqref{eq_MA21} one observes that the remainder term
$K^0\in \B(\Hrond^0)$ of \eqref{eq_formula} is of the form (with $\cong$ meaning
modulo $\K(\Hrond^0)$)
\begin{align*}
K^0 & \cong  (\V_{N/2}^0)^*  I\overline{\psi}(D) m_{N/2,N}(X) \V_N^0  
\otimes \tfrac{1}{N}|\xi^0_{N/2}\rangle \langle \xi^0_N| \\
& \quad -  (\V_N^0)^*\overline{\psi}(D) m_{N,N/2}(X) I \V_{N/2}^0 
\otimes \tfrac{1}{N}|\xi^0_{N}\rangle \langle \xi^0_{N/2}|  \\
& =  (\V_{N/2}^0)^*  I  \V_N^0\overline{\psi}(D^0_N) m_{N/2,N}\big(\arctanh(\tfrac{X^0-2}{2})\big) \otimes \tfrac{1}{N}|\xi^0_{N/2}\rangle \langle \xi^0_N| \\ 
& \quad - \overline{\psi}(D^0_N) m_{N,N/2}\big(\arctanh(\tfrac{X^0-2}{2})\big) (\V_N^0)^* I \V_{N/2}^0 \otimes \tfrac{1}{N}|\xi^0_{N}\rangle \langle \xi^0_{N/2}|  \\
& =  \iota \overline{\psi}(D^0_N) m_{N/2,N}\big(\arctanh(\tfrac{X^0-2}{2})\big) 
\otimes \tfrac{1}{N}|\xi^0_{N/2}\rangle \langle \xi^0_N|  \\
& \quad - \overline{\psi}(D^0_N) m_{N,N/2}\big(\arctanh(\tfrac{X^0-2}{2})\big) \iota^* 
\otimes \tfrac{1}{N}|\xi^0_{N}\rangle \langle \xi^0_{N/2}|
\end{align*}
with the operator $\iota\in \B\big(\ltwo(0,4),\ltwo(-4,0)\big)$ defined in \eqref{eq_iota}. In addition, the function $\psi$ belongs to $C_0(\R)$, while the functions
$ m_{N/2,N}\big(\arctanh(\tfrac{\cdot-2}{2})\big) $ and  $m_{N,N/2}\big(\arctanh(\tfrac{\cdot-2}{2})\big) $ belong to $C_0\big([0,4)\big)$,
namely they are continuous, vanish at $4$, and converge to a constant at $0$. 
Thus, let us introduce a subalgebra of $\B(\Hrond^0)$ which is going to contain
the new terms. Note that this algebra could have been already introduced  
in Section \ref{sec_Cordes_alg}.

Let us set
\begin{align*}
&\EE^0_2\\
& := C^*\bigg(\iota \psi(D^0_N)\varphi(X^0)\otimes \tfrac{1}{N}|\xi^0_{N/2}\rangle \langle \xi^0_N|, \psi(D^0_N)\varphi(X^0)\iota^* \otimes \tfrac{1}{N}|\xi^0_{N}\rangle \langle \xi^0_{N/2}|, \\
&\qquad \quad \eta(D^0_N)\varphi_N(X^0)\otimes \tfrac{1}{N}|\xi^0_{N}\rangle \langle \xi^0_N|, \eta(D^0_{N/2})\varphi_{N/2}(X^0)\iota^* \otimes \tfrac{1}{N}|\xi^0_{N/2}\rangle \langle \xi^0_{N/2}|,   \hbox{with }\\
& \qquad \quad\psi \in C_0(\R),
\eta \in C_0\big([-\infty,\infty)\big),
 \varphi \in C_0\big([0,4)\big),
\varphi_N \in C_0\big([0,4]\big),\varphi_{N/2} \in C_0\big([-4,0]\big)
\bigg)^+.
\end{align*}
Note that here, the exponent $+$ means that $\C$ times identity operator on the subspace
\begin{equation}\label{eq_++}
\ltwo\big((-4,0);\P^0_{N/2} \C^N\big) \, \oplus  \, \ltwo\big((0,4);\P^0_N \C^N\big)
\end{equation}
has been added to the algebra. 
Clearly, this $C^*$-algebra is the smallest one which contains  the restriction of $\F^0 W^0_-(\F^0)^*$ to the subspace defined in \eqref{eq_++},
see also Proposition \ref{prop_restriction}. 
It is also this algebra which can be mapped into a subalgebra of $M_4(\EEC)$ (after one identification and two unitary transforms, as considered in the previous subsections).

Based on this algebra, we can now provide the definition of an algebra which is an enlarged version of the algebra $\EE^0$ defined in \eqref{eq_def_EE} for $\theta=0$. 
Indeed, we can set
\begin{align*}
\EE^0_d:=& C^*\big(\EE^0, \EE^0_2\big)
\end{align*}
which is a unital $C^*$-subalgebra of $\B(\Hrond^0)$ containing $\EE^0$. 
By construction, it also contains the operator $K^0$  of \eqref{eq_formula}, which means
that it contains the operator $\F^0\big(W_-^0-1\big)(\F^0)^*$
even in the intricate case.

The main result about this algebra is related to product of its two constituents.
For that purpose, we recall from \eqref{eq:sum} that, modulo compact terms, a generic element of $\EE^0$ can be written as
\begin{equation}
\eta(D^0)\;\!a
=  \sum_{j,j'}
\big(\eta(D^0_{j})\;\!\a_{j,j'}\big)\otimes \;\!\tfrac{1}{N}|\xi^0_j\rangle \langle \xi^0_{j'}|
\end{equation}
with $\a_{j,j'}$ in $C_0\big(I_j^\theta \cap I_{j'}^\theta; \C\big)$ if  
$\lambda_j^\theta\neq \lambda_{j'}^\theta$, and in $C\big(\overline{I_j^\theta}; \C\big)$ if $\lambda_j^\theta = \lambda_{j'}^\theta$, see Remark \ref{rem:misleading}.
The function $\eta$ belongs to $C_0\big([-\infty,\infty)\big)$.

\begin{Lemma}\label{lem_product}
For any $\psi \in C_0(\R)$ and $\varphi \in C_0\big([0,4)\big)$, the product between an element of $\EE^0$
and the operator $\iota \psi(D^0_N)\varphi(X^0)\otimes \tfrac{1}{N}|\xi^0_{N/2}\rangle \langle \xi^0_N|$ or the operator $\psi(D^0_N)\varphi(X^0)\iota^* \otimes \tfrac{1}{N}|\xi^0_{N}\rangle \langle \xi^0_{N/2}|$
belong to $\EE^0_2 + \K(\Hrond^0)$.
\end{Lemma}

\begin{proof}
This easy proof is omitted but is based on the following observation: For any 
$\varphi_+\in C_0\big((0,4]\big)$ and $\varphi_-\in C_0\big([-4,0)\big)$ the operators
$$
\iota \psi(D^0_N)\varphi(X^0) \varphi_+(X^0)  \qquad \hbox{ and }\qquad 
\varphi_-(X^0)\iota \psi(D^0_N)\varphi(X^0)
$$
belong to $\K\big(\ltwo(0,4), \ltwo(-4,0)\big)$, and similarly
the operators 
$$
\varphi_+(X^0) \psi(D^0_N)\varphi(X^0)\iota^*
\qquad \hbox{ and }\qquad 
\psi(D^0_N)\varphi(X^0)\iota^*\varphi_-(X^0)
$$
belong to $\K\big(\ltwo(-4,0), \ltwo(0,4)\big)$.
\end{proof}

For the other two generators of the algebra $\EE^0_2$, observe that they also belong to the algebra $\EE^0$, meaning that their products with elements of $\EE^0$ belong to  $\EE^0$. This observation together with the content of Lemma \ref{lem_product}
has an important consequence: the quotient of the algebra $\EE^0_d$ by $\K(\Hrond^0)$
is given by the union of the two quotients. Since both are already described in the previous
sections, it will be rather easy to get the topological version of Levinson's theorem,
once we get for arbitrary $N\in 2\N$ results similar to the ones contained in Lemmas \ref{lem_aa} and \ref{lem_bb}. 
In fact, one checks that the proof of Lemma \ref{lem_aa} holds for any $N$. For that reason, we only provide the updated statement.

\begin{Lemma}
When $\v\xi^0_{N/2}$ and $\v\xi^0_{N}$ are linearly independent, 
 $\s^0_{N/2,N/2}(0_-) = -1$ and $\s^0_{N,N}(0_+) =  -1$. 
When there exists $\alpha\in \{-1,1\}$ satisfying $\v\xi^0_{N/2}=\alpha \v \xi^0_N$,
one has $\s^0_{N/2,N/2}(0_-) = -i$ and $\s^0_{N,N}(0_+) =  i$.
\end{Lemma} 

On the other hand, Lemma \ref{lem_bb} does not seem to hold in the general case
$N\in 2\N$. Nevertheless, the values taken by $\s^0_{N/2,N/2}(-4)$ and by 
$\s^0_{N,N}(4)$ are not arbitrary, as shown in the following statement. 
Note that in Lemma \ref{lem:11}, a similar result was based on the unitarity
of the image of the wave operator in the quotient algebra, while in Lemma \ref{lem_bb}
we used the explicit form of the scattering operator at thresholds. 
Here, we shall use that the determinant of the function $\Gamma^4$
introduced in Section \ref{sec_Cordes_alg} has to be a complex number of modulus $1$.

\begin{Lemma}
For $N\in 2\N$ one has $\s^0_{N/2,N/2}(-4)\in \{-1,1\}$ and 
$\s^0_{N,N}(4)\in \{-1,1\}$.
\end{Lemma}

\begin{proof}
The case when $\v\xi^0_{N/2}$ and $\v\xi^0_{N}$ are linearly independent
has already been treated in Lemma \ref{lem:11}, since this corresponds to the non-intricate case.
Thus, let us assume that  there exists $\alpha\in \{-1,1\}$ satisfying $\v\xi^0_{N/2}=\alpha \v \xi^0_N$. Let us also look back at the expression obtained for $\Gamma^4$ at the end of Section \ref{sec_Cordes_alg}. These expressions are still valid in the case $N\in 2\N$
since they are computed on the part of the quotient algebra given by $\EE^0_2$.

For $\alpha^2=1$,  $\s^0_{N/2,N/2}(0_-) = -i$ and $\s^0_{N,N}(0_+) =  i$, one easily obtains that the following equality holds
$$
4cd+(b+a)(f-e)=-i\eta_+(\xi),
$$
where the l.h.s.~is part of the expression in \eqref{eq_4mat}.
Similarly, if we set $\s^0_{N/2,N/2}(-4)=:\e^{i\alpha}$ and $\s^0_{N,N}(4)=:\e^{i\beta}$, then the expression for $(b-a)(f+e)$ is given by
$$
\frac{1}{4}\Big( (1-\e^{i\alpha})(1-\e^{i\beta})\eta_+(\xi)^2+
2(\e^{i\beta}-\e^{i\alpha})\eta_+(\xi) -(1+\e^{i\alpha})(1+\e^{i\beta})\Big).
$$
In order to keep a unitary valued function, and since the three functions
$\xi\mapsto \eta_+(\xi)^2$, $\xi \mapsto \eta_+(\xi)$, and $\xi \mapsto 1$
are unitary valued and independent, it is necessary and sufficient to keep only 
one of them with non-trivial coefficients. 
By imposing this conditions, one directly obtain the possible values for $\s^0_{N/2,N/2}(-4)$ and for $\s^0_{N,N}(4)$.
\end{proof}

By summing up the four alternatives of the previous result, and by collecting similar information for $\Gamma^2$ one gets
\begin{enumerate}
\item[$(i)$]  If $\s^0_{N/2,N/2}(-4)=1$, $\s^0_{N,N}(4)=1$, then $\det(\Gamma^4)(\xi)=i\eta_+(\xi)$ and $\det(\Gamma^2)(\xi)=-i\eta_-(\xi)$,
\item[$(ii)$] If $\s^0_{N/2,N/2}(-4)=1$, $\s^0_{N,N}(4)=-1$, then $\det(\Gamma^4)(\xi)=i\eta_+(\xi)^2$ and $\det(\Gamma^2)(\xi)=i\eta_-(\xi)^2$,
\item[$(iii)$] If $\s^0_{N/2,N/2}(-4)=-1$, $\s^0_{N,N}(4)=1$, then $\det(\Gamma^4)(\xi)=-i\eta_+(\xi)^2$ and $\det(\Gamma^2)(\xi)=-i\eta_-(\xi)^2$,
\item[$(iv)$] If $\s^0_{N/2,N/2}(-4)=-1$, $\s^0_{N,N}(4)=-1$, then $\det(\Gamma^4)(\xi)=-i\eta_+(\xi)^3$ and $\det(\Gamma^2)(\xi)=i\eta_-(\xi)^3$.
\end{enumerate}
Note that the corresponding function $\Gamma^3$ can be obtained by evaluating these expressions at $\xi=0$.
It only remains to add the information obtained for the algebra $\EE^0$ with the one obtained above to get the final result about Levinson's theorem in the intricate case and for arbitrary $N\in 2\N$.

\begin{Theorem}
In the intricate case, 
the following equality holds:
\begin{equation*}
\Var \big(\lambda \to \det S^0(\lambda)\big) + N - \tfrac{1}{2}C-\tfrac{1}{2}
 = \# \sigma_{\rm p}(H^0),
\end{equation*}
where 
$$
C:=\#\big\{k\mid \tilde \s^0_{k}(\tilde \lambda^0_k\pm 2)=1 \big\} 
+ 2\#\big\{k\mid \tilde \s^0_{k}(\tilde \lambda^0_k\pm 2)= \left(\begin{smallmatrix}
1& 0 \\  0 & 1
\end{smallmatrix}\right)\big\}
+\#\Big\{k\mid \tilde \s^0_{k}(\tilde \lambda^0_k\pm 2)=\left(\begin{smallmatrix}
a& b \\  \overline{b} & -a \end{smallmatrix}\right) \Big\}.
$$
\end{Theorem}

\begin{proof}
As mentioned before the statement, the proof consists simply in collecting the various results obtained so for.
For all thresholds different from $-4$, $0$, and $4$, the content of Theorem \ref{thm:Adam} is still valid, but it does not apply to these three thresholds.
For these three thresholds, which are related to the algebra $\EE_2^0$, the approach
is similar to the one used in Theorem \ref{thm_N=2}, see also Remark \ref{rem_N=2}, 
but with the four cases mentioned above.
\end{proof}

\end{document}